\documentclass[journal,12pt,onecolumn,draftclsnofoot,]{IEEEtran}
\usepackage[utf8]{inputenc}
\usepackage{xcolor}
\usepackage{acronym}
\usepackage{cite}
\usepackage{amsmath,amssymb}
\usepackage{graphicx}
\usepackage{subfigure}
\usepackage{algorithm,algcompatible}
\usepackage{amsthm}
\usepackage{etoolbox}
\usepackage{wrapfig}

\theoremstyle{remark}
\newtheorem{proposition}{Proposition}

\algdef{SE}[SUBALG]{Indent}{EndIndent}{}{\algorithmicend\ }
\algtext*{Indent}
\algtext*{EndIndent}

\begin{document}

\title{Slow Beam Steering and NOMA for Indoor Multi-User Visible Light Communications}

\author{\IEEEauthorblockN{Yusuf Said Ero\u{g}lu,~\IEEEmembership{Member,~IEEE}, Chethan Kumar Anjinappa,\\ \.{I}smail G\"{u}ven\c{c},~\IEEEmembership{Senior~Member,~IEEE}, and Nezih Pala}
\thanks{This work is supported in part by NSF CNS award 1422062. 
Yusuf Said Ero\u{g}lu, Chethan Kumar Anjinappa, and \.{I}smail G\"{u}ven\c{c} are with the Department of Electrical and Computer Engineering, North Carolina State University, Raleigh, NC (e-mail:~\{yeroglu, canjina, iguvenc\}@ncsu.edu). 
Nezih Pala is with the Department of Electrical and Computer Engineering, Florida International University, Miami, FL (e-mail:~npala@fiu.edu). }
}

\maketitle

\vspace{-4mm}
 
\begin{abstract}\vspace{-1mm}
Visible light communication (VLC) is an emerging technology that enables broadband data rates using the visible spectrum. In this paper, considering slow beam steering where VLC beam directions are assumed to be fixed during a transmission frame, we find the steering angles that simultaneously serve multiple users within the frame duration and maximize the data rates. This is achieved by solving a non-convex optimization problem using a grid-based search and majorization-minimization (MM) procedure. Subsequently, we consider multiple steerable beams with a larger number of users in the network and propose an algorithm to cluster users and serve each cluster with a separate beam. We optimize the transmit power of each beam to maximize the data rates. Finally, we propose a non-orthogonal multiple access (NOMA) scheme for the beam steering and user clustering scenario, to further increase the data rates of the users. The simulation results show that the proposed beam steering method can efficiently serve a high number of users, and with the power optimization, a data rate gain up to ten times is possible. The simulation results for NOMA suggests an additional 10 Mbps sum rate gain for each NOMA user pair.
\end{abstract}

\vspace{-1mm}


\vspace{-3mm}
\section{Introduction}
Visible light communication (VLC) technology uses light sources such as LEDs for both illumination and wireless data transfer. In this technology, light-emitting diodes (LEDs) act as an antenna and transmit data to users through modulating light intensity. Due to the high frequency of the modulation, the changes in the signal are not perceivable to the human eye\cite{rajagopal2012ieee}. Depending on the LED or lens type, VLC light beams can be highly directional\cite{Eroglu_JSAC, Nakagawa}. Such directional LEDs can be preferred for providing higher signal strength at longer distances, decreasing interference in other directions, or providing an accurate angle of arrival information for localization purposes. 

VLC networks can provide highly accurate localization information \cite{Alphan_JLT, eroglu_2015}, and this location information can be used to steer the light beam towards user location by manipulating the orientation of the light source to further enhance the communications performance. It has been shown in the literature that using a \textit{steerable directional beam} maximizes both the overall signal strength and the coverage area \cite{Rahaim2017}. In\cite{eroglu_VLCS}, tracking users by steering LEDs is shown to provide a higher signal to interference plus noise ratio (SINR) in the VLC cell borders, which provides smoother handovers between adjacent VLC access points (APs). A beam steering scheme is studied with angle diversity receivers in~\cite{5770677}, where the beam can be steered in some certain orientations which are predetermined depending on the user location distribution. The study is extended for imaging receivers in \cite{alsaadi2010high, hussein201520}. However, these studies assume that each user is tracked with a dedicated LED or multiple LEDs. When the number of users is lower than or equal to the number of steerable beams the steering is relatively simple because each user can be assigned a single beam that tracks the user. However, in some cases, the number of users can be higher than the number of steerable beams. In such cases, how to steer the LEDs and distribute time allocation to users is an open problem which has not been addressed in the literature.

In recent years, non-orthogonal multiple access (NOMA) schemes have received significant attention for cellular networks \cite{6692652, 6868214}. The primary reason for adopting NOMA is its ability to serve multiple users using the same time and frequency resources. NOMA achieves this by assigning different power levels to users that have distinctive channel gains. In \cite{marshoud2016non}, the use of NOMA is investigated for VLC, and it was found that NOMA can serve multiple users to provide higher data rates compared to orthogonal multiple access (OMA) such as time or frequency division. In \cite{haas_noma}, VLC NOMA is studied for two users case, and it is shown that the gain of NOMA over OMA further increases when users with more distinctive channel gains are paired. In \cite{7792590}, NOMA user selection and power allocation are studied, and the power coefficients are derived considering fairness among users in \cite{7752879}. However, use of NOMA has not been addressed in the VLC literature for a beam steering scenario, and it has not been studied considering the inter-beam interference caused from other steerable beams.

In this paper, which is substantially extended from \cite{eroglu_spawc}, we investigate the optimal beam steering parameters for proportionally fair rate allocation, especially for the case where the number of users is higher than the number of steerable beams. The contributions of the paper can be summarized as follows:
\begin{itemize}
\item[i.] We define the steering problem for a single beam and multiple users. The optimization parameters are the steering angles, the directivity index of the beam, and the time allocation of each user. We propose a solution for the non-convex problem using a grid search based optimization and majorization-minimization (MM) procedure. Our results show that the proposed beam steering improves the data rates significantly by increasing the users' signal strength. While the data rate gain can be more than four times with a single user, a higher number of users can also be served by a single beam with a lower data rate gains.
\item[ii.] We propose a method for decreasing the search space to reduce the computation time for the mentioned problem. 
\item[iii.] We evaluate the case where there are multiple steerable beams. As a solution for steering and multiple access in this scenario, we propose a user grouping algorithm which is an extension of the $k$-means clustering algorithm. In particular, we cluster the users and assign a single beam to each cluster. With this method, the time allocation of each user is increased by exploiting the spatial diversity of the users. The simulation results show that ten users can be best served with three independent beams, and the data rate gain due to steering is four times for this case. 
\item[iv.] We find the optimum transmit power of each beam with respect to a total power constraint. We do it by solving a maximization problem that finds the transmit powers that maximize the sum rate or provide proportionally fair rates. The power optimization provides an additional sum rate gain between 30 - 70 Mbps, where the total gain over no steering scheme can be up to 10 times. 
\item[v.] Finally, we propose a NOMA scheme by coupling users in the same cluster to further improve the data rates. We find the optimum NOMA power coefficients for a user pair again utilizing the MM procedure. The MM procedure has not been utilized to achieve the VLC NOMA coefficients in the literature. With the coefficients found by our method, the user pair has a 10 Mbps sum rate gain in a  proportionally fair allocation, where the weaker user gets a larger portion of the gain. 
\end{itemize}

The remainder of the paper is organized as follows. In Section~\ref{SystemModel}, we review beam steering mechanisms and channel model assumptions for VLC and introduce the multi-user beam steering problem. In Section~\ref{PropSol}, we first consider that there is a single steerable beam, and present the solution to the introduced problem using grid search and MM procedure. In Section~\ref{MultipleBeam}, we extend the solution to multiple steerable beams case where users are clustered and a single beam is assigned to each cluster. In Section~\ref{NOMAsection}, we propose a NOMA scheme for the users in the same cluster. In Section~\ref{SimRes}, we present the simulation results, and finally, we conclude the paper in Section~\ref{Conc}.
\textit{Notations:} The Euclidean norm is denoted by $||.||_2$, the transpose of a vector/matrix is denoted by $[.]^T$, and vectors are represented by bold symbols. 

\vspace{-3mm}

\section{System Model} \label{SystemModel}
In this section, we investigate the approaches proposed for VLC beam steering and introduce the slow beam steering problem.

\vspace{-3mm}

\subsection{VLC Beam Steering Model and Assumptions}
Piezoelectric beam steering is proposed in \cite{eroglu_VLCS} in order to track the user, improve the signal strength, and provide smoother handover between different VLC APs. Piezo actuators convert electrical signal into precisely controlled physical displacement. This property of piezo actuators is used to finely adjust machining tools, camera lenses, mirrors, or other equipment \cite{piezoactuators}. Piezoelectric actuators can also be used to tilt LEDs or lenses, to steer the beam directed towards user location. In Fig.~\ref{Piezo}, two different beam steering schemes using piezo actuators are illustrated. In Fig.~\ref{Piezo}(a), whole LED is tilted using a set of piezo actuators, while in Fig.~\ref{Piezo}(b), only the lens is steered. The setup in Fig.~\ref{Piezo}(b) makes it possible to change the directivity of the light beam by shifting the lens forward or backward. In order to tilt an LED to any angle, two sets of piezo actuators can be used: while one provides steering on one direction, the other provides steering in a perpendicular direction.

\begin{figure}[tb]
	\centering\vspace{-1mm}
	\includegraphics[width = 3.2 in]{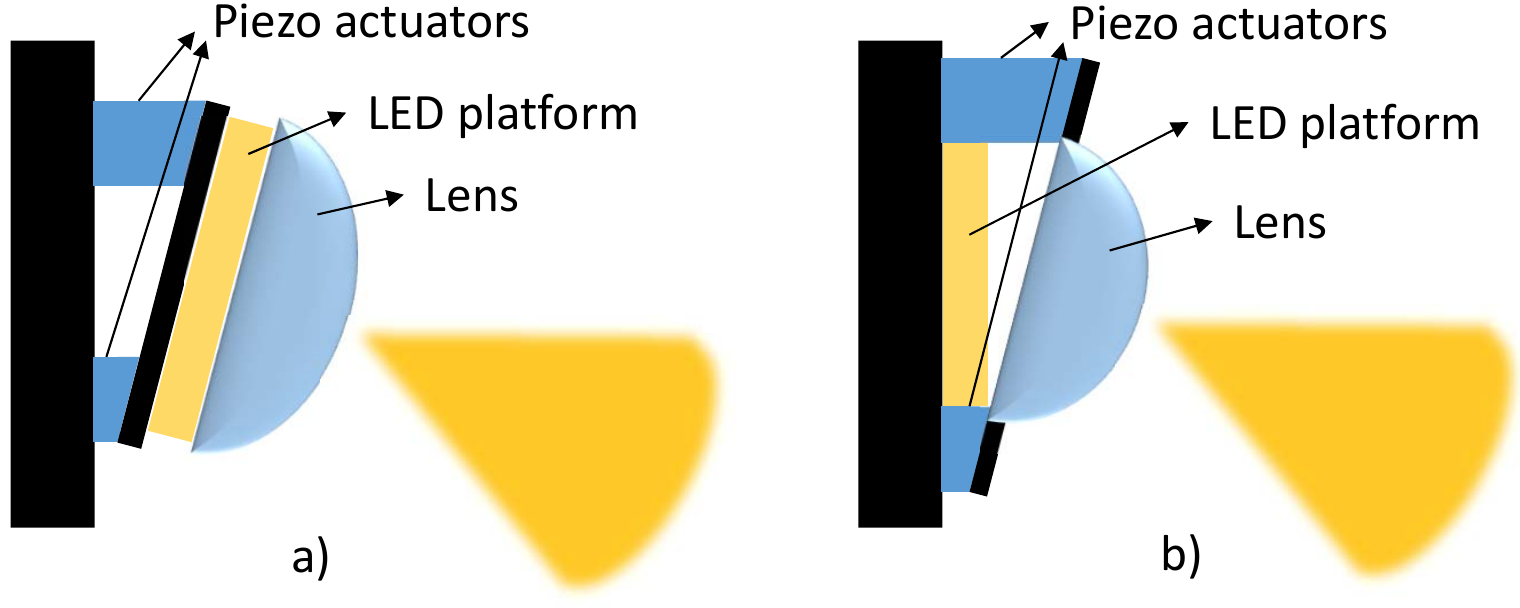}\vspace{-5mm}
	\caption{VLC beam steering using piezo actuators. (a) The LED and the lens are steered together. (b) Only the lens is steered.}\vspace{-6mm}
	\label{Piezo}
\end{figure}

Another method to steer LED light is to use micro-electro-mechanical system (MEMS) based mirrors \cite{Rahaim2017, 6525314, Morrison:15}, where the direction of the beam is controlled by changing the orientation of micromirrors. In \cite{Morrison:15}, a setup with LEDs and MEMS mirrors is presented with steering angles of $\pm 40^\circ$ with a settling time under 5~ms, additionally featuring adaptable beam directivity. As a similar method to MEMS mirrors steering, in \cite{6888049}, optical gratings are used to change the beam direction. MEMS mirrors are also studied in the context of steering laser beams for indoor free space optical (FSO) communications~\cite{Varifocal, Oh:15, Knoernschild}. The phased arrays are also used for beam steering/beam forming of optical wireless signals\cite{Wu2014, resler1996high}.  

In this study, without explicitly assuming any of the aforementioned beam steering methods, we consider a VLC AP with a limited number of steerable beams that can be steered within a given range. Additionally, we consider two scenarios where: 1) the beam directivity is fixed, or 2) the beam directivity can be changed within a given range.

\vspace{-3mm}

\subsection{Channel Model}
Initially, we consider an AP with a single steerable light beam and $K$ users, and Fig.~\ref{SingleBeam} shows an example scenario for $K = 3$. 
\begin{figure}[tp]
	\centering
	\vspace{1mm}
	\subfigure[Single steerable beam.]{
		\includegraphics[width=1.29in] {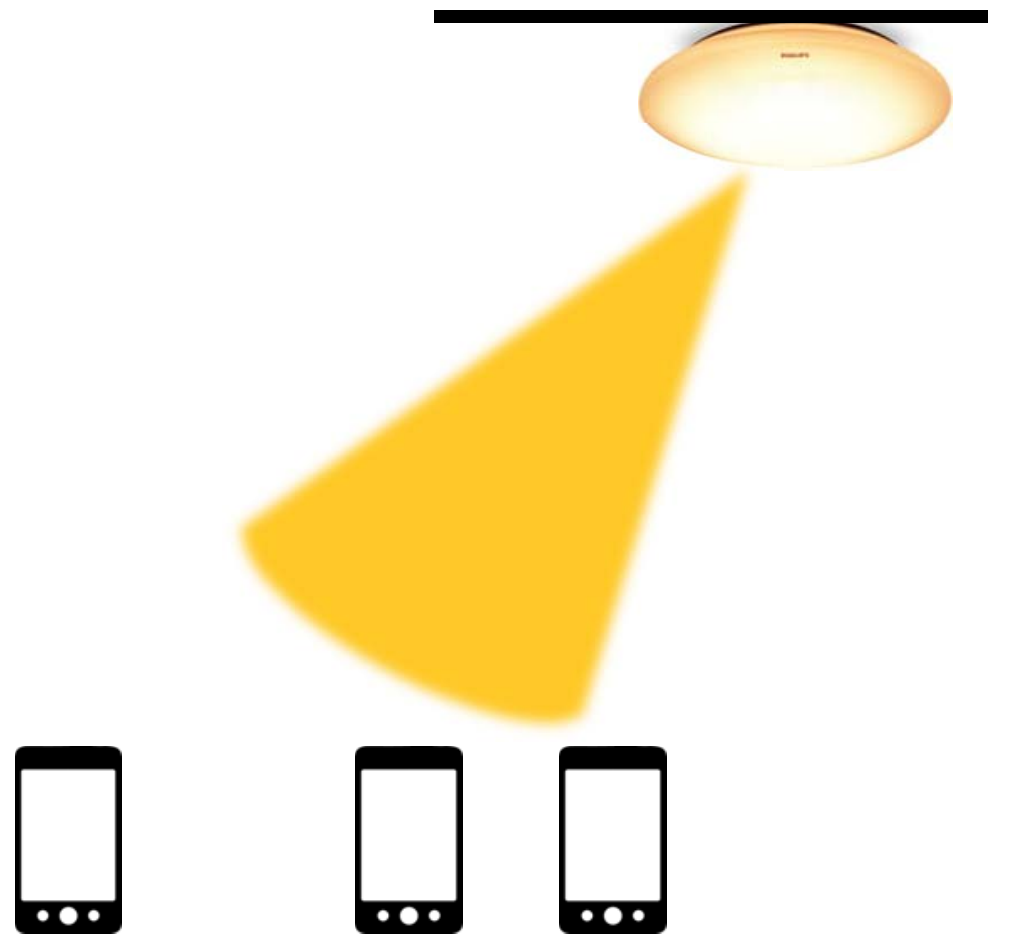}
		\label{SingleBeam}
	}
	\subfigure[Multiple steerable beams.]{
		\includegraphics[width=1.95in] {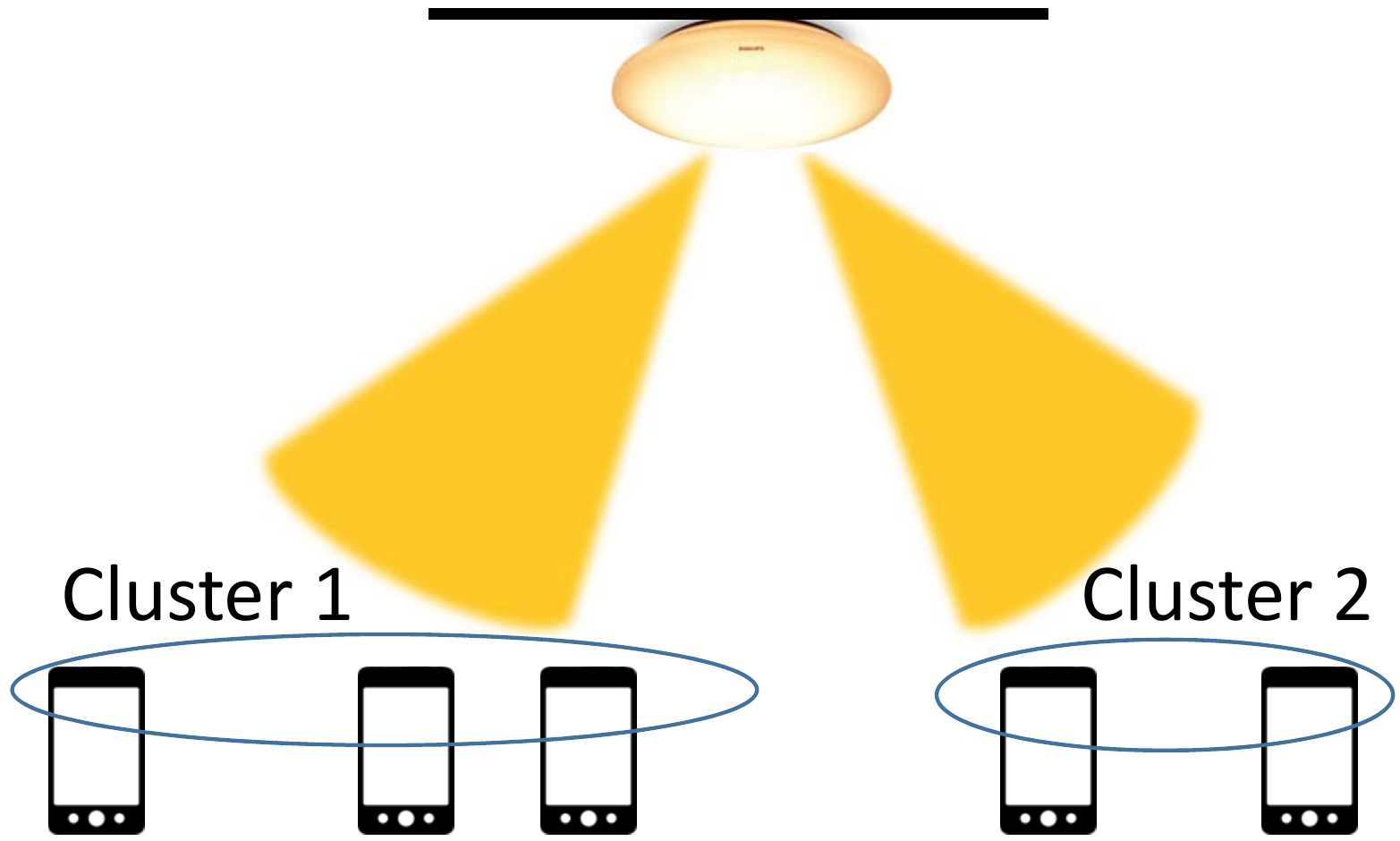}
		\label{MultiBeam}
	}\vspace{-1mm}
    \caption{Steering single and multiple VLC beams to user clusters. }\vspace{-6mm}
	\label{SteeringScenarios}
\end{figure}
The AP serves all users with time division multiple access (TDMA), and $k$-th user is served with time ratio $\tau_k$, where $0 \leq \tau_k \leq 1$. We aim at finding the steering angles and LED directivity index which maximizes logarithmic sum rate of all users. In 3D model, we need two angles to specify the orientation of the beam, which are the elevation and the azimuth angles, denoted by $\alpha$ and $\beta$, respectively, as shown in Fig.~\ref{SteeringAngles}. We can convert these angles to an orientation vector given as:
\begin{align}
\textbf{n}_{\rm tx}&= [n_{x\rm(tx)}, n_{y\rm(tx)}, n_{z\rm(tx)}]^T =[\cos(\beta)\cos(\alpha),~ \sin(\beta)\cos(\alpha),~ \sin(\alpha)]^T. \label{orient}
\end{align}
The location of the AP is $\textbf{r}_{\rm tx}=[x_{\rm tx}, y_{\rm tx}, z_{\rm tx}]^T$. Likewise, the location and the orientation of the $k$-th user are $\textbf{r}_{k}= [x_{k}, y_{k}, z_{k}]^T$, and $\textbf{n}_{\rm k}= [n_{x(k)}, n_{y(k)}, n_{z(k)}]^T$, respectively. Then, the vector from the AP to $k$-th user is  $\textbf{v}_k = \textbf{r}_{k} - \textbf{r}_{\rm tx} = [v_{x(k)}, v_{y(k)}, v_{z(k)}]^T.$  The distance between the LED and the $k$-th user is $d_k = ||\textbf{v}_k||_2$, while the angle between the LED orientation and $\textbf{v}_k$ is denoted as $\phi_k$, and we can write:
\begin{align}
\cos(\phi_k) = \frac{\textbf{n}_{\rm tx}^T(\textbf{r}_k-\textbf{r}_{\rm tx})}{d_k} = \frac{\textbf{v}_k^T \textbf{n}_{\rm tx}}{||\textbf{v}_k||_2} \, . \label{cosPhi}
\end{align}
Similarly, the angle between the receiver orientation and $\textbf{v}_k$ is $\theta_k$, and we can write:
\begin{align}
\cos(\theta_k) = \frac{\textbf{n}_k^T(\textbf{r}_{\rm tx}-\textbf{r}_k)}{d_k} = - \frac{\textbf{v}_k^T \textbf{n}_{k}}{||\textbf{v}_k||_2} \, . \label{cosTheta}
\end{align}
We assume a light beam radiation follows the Lambertian pattern \cite{219552}, with $\gamma$ being the directivity index of the beam. The effect of the directivity index on the beam shape is illustrated in Fig.~\ref{BeamDirectivity}, where two contours receive the same power from two LEDs with different $\gamma$. Then, assuming the receiver has a wide field of view (FOV), we can remove the FOV constraint, and the line-of-sight (LOS) channel gain of the $k$-th user can be calculated using \eqref{orient}-\eqref{cosTheta} as follows:
\begin{align}
h_k =~ \frac{\gamma + 1}{2\pi}& A_rr\cos^\gamma(\phi_k)\cos(\theta_k)\frac{1}{d_k^2} \label{h_first}
~=~ \frac{\gamma + 1}{2\pi} A_rr \frac{ (\textbf{v}_k^T\textbf{n}_{\rm tx})^\gamma \, \textbf{v}_k^T\textbf{n}_k }{||\textbf{v}_k||_2^{n+3}} \\ 
=~ \frac{\gamma + 1}{2\pi}& A_rr\frac{\left( v_{x(k)}n_{x(k)} + v_{y(k)}n_{y(k)} + v_{z(k)}n_{z(k)}\right)}{\left(v_{x(k)}^2 + v_{y(k)}^2 + v_{z(k)}^2\right)^\frac{\gamma+3}{2}} \label{channelGain}
\\ &\times \big(v_{x(k)}\cos(\beta)\cos(\alpha) + v_{y(k)}\sin(\beta)\cos(\alpha) + v_{z(k)}\sin(\alpha) \big)^\gamma, \nonumber
\end{align}
where $A_r$ is the detection area of the PD, and r is the responsivity coefficient. Using \eqref{channelGain}, the rate of the $k$-th user is given as~\cite{6883844}
\begin{align}
R_k = B\log\left(1 + \frac{(ph_k)^2}{N_0B}\right),
\end{align}
where $p$ is the transmit power of the LED. The $N_0$ is the spectral density of additive white Gaussian noise, and $B$ is the communication bandwidth. 

\begin{figure}[tp]
	\centering
	\vspace{-1mm}
	\subfigure[The elevation and azimuth angles, $\alpha$ and $\beta$, respectively.]{
		\includegraphics[width=1.5in] {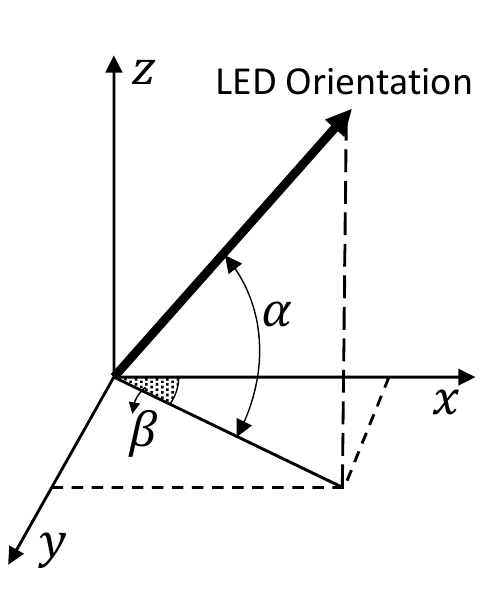}
		\label{SteeringAngles}
	}~
	\subfigure[Equal power contours for different directivity index, $\gamma$.]{
		\includegraphics[width=1.6in] {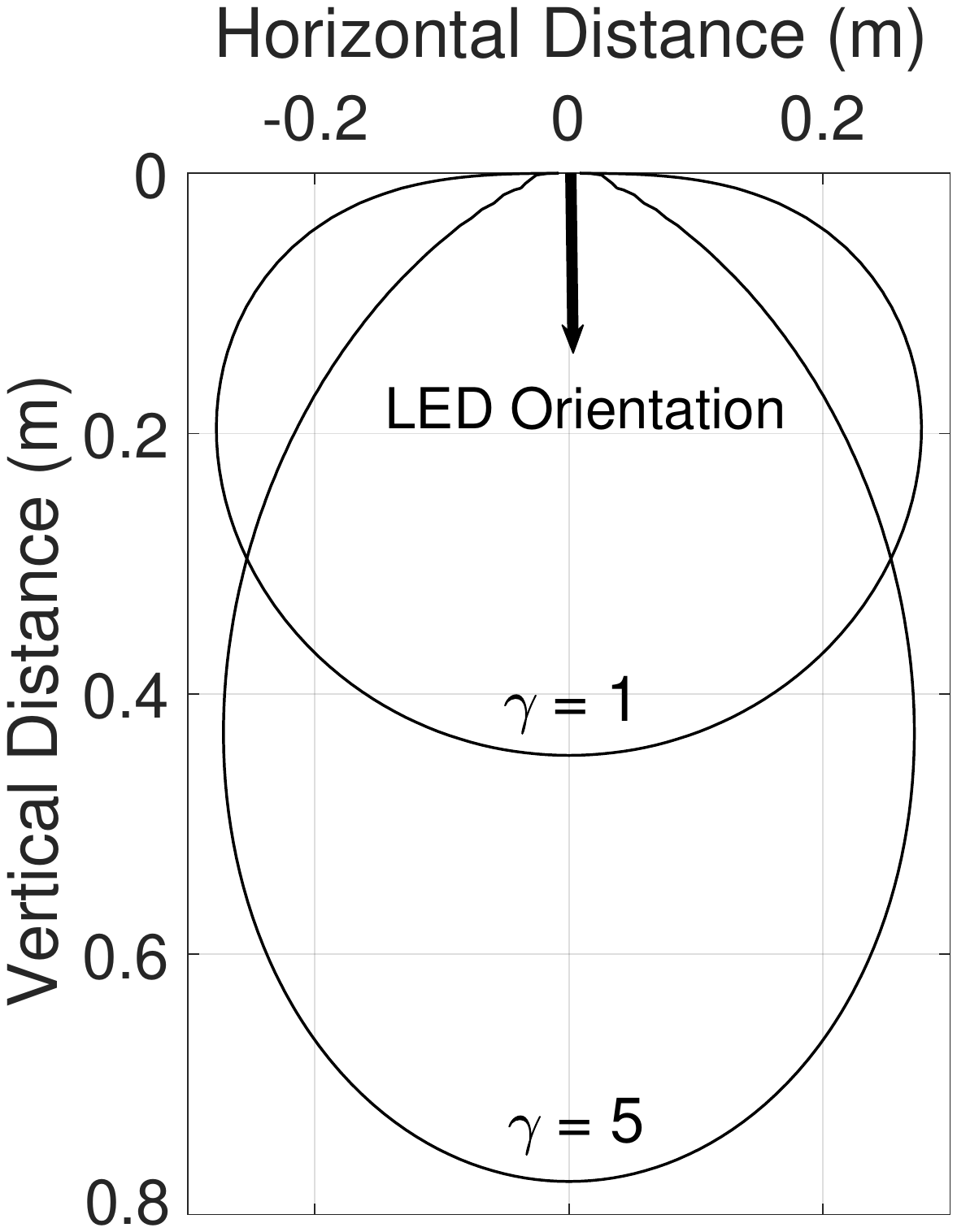}
		\label{BeamDirectivity}
	}\vspace{-1mm}
    \caption{Illustration of steering angles and directivity index. }
	\label{AlfaBetaGama} \vspace{-6mm}
\end{figure}

\vspace{-3mm}

\subsection{Slow Beam Steering for Multiple Access VLC}
We consider a model where the beam is steered so that multiple users can access the channel with TDMA without changing the beam orientation towards each user. There are two reasons not to consider changing beam orientation each time slot for each user. The first one is, there will be time loss between each time slot for orientation change. The shortest reported settling time for LED beam steering is 5~ms~\cite{Morrison:15}, which is close to the whole TDMA frame length used for Wi-Fi systems. The second reason is that it is not possible to do such a switching without a flickering effect. The human eye can capture changes up to 200 Hz \cite{rajagopal2012ieee}, which means the the periodic changes to the signal should settle under 5 ms. Considering that just one steering takes around 5 ms, it is not possible to quickly switch the beam between users without flickering. 

In this paper, we propose a solution where the beam is steered once for a given set of user locations, and no more steering is needed unless the location and orientation of the users change significantly. If any user movement occurs, new steering parameters are computed and the \emph{slow} beam is steering is carried out. Accordingly, the steering parameters can be found by solving the following constrained optimization problem:
\begin{equation}
\begin{split}
\tilde{\boldsymbol{\tau}}, \tilde{\alpha}, \tilde{\beta}, \tilde{\gamma} = {\rm arg}&\underset{ \boldsymbol{\tau}, \alpha, \beta, \gamma}{\rm ~max}~~  \displaystyle\sum_{k=1}^{K} \log(\tau_kR_{k}),\\
{\rm s.t.}~~ c_1:\quad & \alpha_{\rm min} \leq \alpha \leq \alpha_{\rm max}, \\
c_2:\quad & 0 \leq \beta \leq 360^\circ \,, \\
c_3:\quad & \gamma_{\rm min} \leq \gamma \leq \gamma_{\rm max} \,, \\ 
c_4:\quad & \sum_{k=1}^{K} \tau_k = 1 \, ,
\end{split}
\label{betaOpt}
\end{equation}
where $\alpha, \beta$ and $\gamma$ are beam steering and directivity parameters as captured in \eqref{orient}, \eqref{channelGain} and illustrated in Fig.~\ref{AlfaBetaGama}. The constraint $c_1$ limits the elevation angle within the steering capacity of the beam. The constraint $c_2$ is the azimuth limit, which shows that the beam can be steered towards any direction as long as the elevation angle allows. The constraint $c_3$ is for the limits of beam directivity index, which is decided by the device capabilities. The $\boldsymbol{\tau} = [\tau_1, ... , \tau_K]$ is the time division coefficient vector whose elements add up to 1. To make sure all users are served and the resources are distributed fairly, the objective function is the sum of logarithmic rate instead of sum rate~\cite{3010473}. If the logarithm is removed from objective function, a single user with the largest channel gain gets all time allocation and the beam is steered towards that user, leaving other users unserved. The solution of \eqref{betaOpt} will be discussed in Section~\ref{PropSol}.

\vspace{-3mm}

\subsection{NOMA Signal Model}
The TDMA serves each user on different time slots. Alternatively, NOMA serves all users simultaneously by exploiting the channel gain differences of users. Let there be $K$ NOMA users served by the same transmitter LED. The users are ordered based on the magnitude of their channel gains so that $h_1 < h_2 < ... < h_K$. The transmitter sends the signal to all users simultaneously by superposing the symbols in the power domain and adding a DC bias. The signal to be transmitted by the LED is
\begin{align}
x = p\sum_{k=1}^{K} \rho_ks_k +  I_{\rm DC}
\end{align}
where $p$ is the transmit power of the LED, $I_{\rm DC}$ is the DC bias added to the signal to ensure positive intensity, $s_k$ is the modulated message symbol for the $k$-th user, and $\rho_k$ is the NOMA power allocation coefficient for the $k$-th user. The message symbol signals are assumed to have zero mean and unit variance. In NOMA, users with poor channel conditions are allocated higher power. Therefore, $\rho_1 \ge \rho_2 \ge ... \ge \rho_K$ to make interference cancellation possible, and $\sum_{k=1}^{K} \rho_k^2 = 1$ to satisfy total electricity power constraint\cite{haas_noma}. Removing the DC bias at the receiver, the remaining received signal at $\ell$-th user is given by
\begin{align}
y_\ell = ph_\ell\sum_{k=1}^{K} \rho_ks_k + z_\ell,
\end{align}
where $z_\ell$ is the real-valued Gaussian noise with zero mean and variance $\sigma_\ell^2$. A constant noise power spectral density $N_0$ is assumed so that $\sigma_\ell^2 = N_0B$. Successive interference cancellation (SIC) is carried out to remove the signals of users with weaker channel gains. This is possible because the NOMA power coefficients of these signals are higher, therefore the symbol can be detected and removed from the received signal, as will be discussed in Section~\ref{NOMAsection}. On the other hand, the signals of stronger users are not canceled and treated as noise.

\vspace{-3mm}

\section{Single Steerable Beam} \label{PropSol}
In this section, we solve the optimization problem in~\eqref{betaOpt} for a single steerable beam and multiple users and introduce a method for decreasing the complexity of the solution. Subsequently, Section~\ref{MultipleBeam} will study the multiple steerable beam scenario. 

\vspace{-3mm}

\subsection{Solution to the Optimization Problem using MM}
We can divide the problem in \eqref{betaOpt} into two independent maximization problems by rewriting the objective function as follows:
\begin{equation}
\displaystyle\sum_{k=1}^{K} \log(\tau_kR_{k}) = \displaystyle\sum_{k=1}^{K} \log(\tau_k) + \displaystyle\sum_{k=1}^{K} \log(R_{k}). \label{sub_problems}
\end{equation}
Then, using the first summation in \eqref{sub_problems}, the first problem in \eqref{betaOpt} becomes: 
\begin{equation}
\tilde{\boldsymbol{\tau}} = {\rm arg}~\underset{ \boldsymbol{\tau}}{\rm max}~ \log \left(\prod_{k=1}^{K} \tau_k\right),  \label{tauOpt}
\end{equation}
subject to only $c_4$ in \eqref{betaOpt}. The answer to this trivial problem is $\tilde{\tau}_k = 1/K,~ \forall k$. The second problem based on \eqref{sub_problems} is given by
\begin{equation}
\tilde{\alpha}, \tilde{\beta}, \tilde{\gamma} = {\rm arg}\underset{ \alpha, \beta, \gamma}{\rm ~max}~~ \displaystyle\sum_{k=1}^{K} \log(R_{k}),
\label{betaOpt2}
\end{equation}
subject to $c_1$, $c_2$, and $c_3$ in \eqref{betaOpt}. The problem in \eqref{betaOpt2} is non-convex, hence gradient based optimization methods get stuck in a local optima. This can be seen in the channel gain in \eqref{channelGain}, which has sine, cosine, and exponential functions of optimization parameters. 

In order to not to use \eqref{channelGain} in the objective function, we follow a grid search based method and calculate the channel gain for discrete values of $\alpha, \beta$, and $\gamma$. To give an example, we separate all available range for $\alpha$ to discrete values with a small sampling interval $\delta$, hence we have $\boldsymbol{\alpha} = [\alpha_{\rm min}, \alpha_{\rm min} + \delta, ... , \alpha_{\rm max}]$. A similar sampling is also used for $\beta$ and $\gamma$, and the sizes of $\boldsymbol{\alpha}, \boldsymbol{\beta}$ and $\boldsymbol{\gamma}$ are $s_\alpha, s_\beta$, and $s_\gamma$, respectively. We calculate the channel gain for all possible $\alpha, \beta$, and $\gamma$ combinations and form a column vector $\textbf{h}_{\boldsymbol{\alpha, \beta, \gamma}}^{(k)}$, whose length is $s_\alpha \times s_\beta \times s_\gamma$, and its indices can be mapped back to $\alpha, \beta$, and $\gamma$. 

As a result, we can restructure the optimization problem in \eqref{betaOpt2} as follows:
\begin{align}
\tilde{\textbf{d}} = {\rm arg}\underset{ \textbf{d}}{\rm ~max} \displaystyle\sum_{k=1}^{K} \log & \left(B\log\left( 1 + \frac{ \left( p\, \textbf{d}^{T}\textbf{h}_{\boldsymbol{\alpha, \beta, \gamma}}^{(k)} \right)^2 }{N_0B}\right) \right), \nonumber \\
{\rm s.t.}~~ c_1 &:\quad \sum_{i = 1}^{s_\alpha s_\beta s_\gamma} d_{i} = 1, \label{betaOpt4} \\ \nonumber
c_2 &: \quad d_{i} = \{0,1\} \quad \forall i, 
\end{align}
where $\textbf{d}$ is a vector same size as $\textbf{h}_{\boldsymbol{\alpha, \beta, \gamma}}^{(k)}$. The constraints $c_1$ and $c_2$ in \eqref{betaOpt4} enforce that only one element of $\textbf{d}$ is equal to one, and the others are all equal to zero. Therefore, the vector multiplication results in choosing an element of $\textbf{h}_{\boldsymbol{\alpha, \beta, \gamma}}$. The problem with \eqref{betaOpt4} is the combinatorial nature of the problem due to the binary constraint $d_{i}$'s. 

In order to remove the binary constraint, we modify the problem further as:
\begin{align}
\tilde{\textbf{d}} = {\rm arg}\underset{ \textbf{d}}{\rm ~max} \displaystyle\sum_{k=1}^{K} \log & \left(B\log\left( 1 + \frac{ \left( p\, \textbf{d}^{T}\textbf{h}_{\boldsymbol{\alpha, \beta, \gamma}}^{(k)} \right)^2 }{N_0B}\right) \right) - \lambda ||\textbf{d}||_{0}, \label{betaOpt5} \\[2pt] 
{\rm s.t.}~~ &c_1:\quad \sum_{i = 1}^{s_\alpha s_\beta s_\gamma} d_{i} = 1;\quad  d_{i} \geq 0  \nonumber
\end{align}
where $||.||_0$ is the $\ell_0$ norm and $\lambda$ is a positive penalty parameter. This modification does not change the meaning of the problem in \eqref{betaOpt4}, however it is still combinatorial due to the $\ell_0$ norm\cite{eroglu_spawc, IRM_Cheth_Chandra}. Note that the problem in (15) and (14) are equivalent in nature because the solution $\tilde{\textbf{d}}$ obtained by solving (14) and (15) would be the same.  The rationale for using the penalized $\ell_0$ norm is that it helps to get rid of the binary constraints and it, along with $c_1$, still preserves the meaning of the problem by forcing the $\textbf{d}$ to be sparse with just one element being one.

The problem in \eqref{betaOpt5} can be relaxed by replacing $\ell_0$ norm with a strictly concave function (e.g. $\ell_q$ norm with $0<q<1$)\cite{MM_Tutorial}. Upon relaxation, the problem turns out to be non-convex and a near-optimal solution can be achieved using the MM procedure \cite{ he2017codebook, sun2017majorization}. The basic idea of the MM procedure is to keep the convex part as it is, and linearize the concave part of the function around a solution obtained in the previous iteration. The relaxed optimization problem with linearized $\ell_q$ norm can be written as follows:
\begin{align}\label{betaOpt6}
\tilde{\textbf{d}} = {\rm arg}\underset{ \textbf{d}}{\rm ~max} \displaystyle\sum_{k=1}^{K} \log & \left(B\log\left( 1 + \frac{ \left( p\, \textbf{d}^{T}\textbf{h}_{\boldsymbol{\alpha, \beta, \gamma}}^{(k)} \right)^2 }{N_0B}\right) \right) - \lambda \sum_{i=1}^{s_\alpha s_\beta s_\gamma}W_i(t) d_i, \\ \nonumber
{\rm s.t.}~~ c_1 &:\quad \sum_{i = 1}^{s_\alpha s_\beta s_\gamma} d_{i} = 1;\quad  d_{i} \geq 0,
\end{align}
where $W_i(t) = q (d_i+\epsilon)^{q-1}$ is the weight update of the majorizer function at iteration $t$, and $\epsilon$ is a small non-negative number added to overcome the singularity issue; without $\epsilon$, $W_i(t)$ becomes undefined at $d_i = 0$. Interested readers may refer \cite{IRM_Cheth_Chandra} and references therein for further details of the MM procedure. Solving (\ref{betaOpt6}) returns $\textbf{d}$, and the index of the non-zero element in $\textbf{d}$ can be mapped back to the optimal $\alpha, \beta$, and $\gamma$ values.

\vspace{-3mm}

\subsection{Decreasing the Steering Angle Search Space}
The optimization in \eqref{betaOpt6} operates over the whole search space in $\textbf{h}_{\boldsymbol{\alpha, \beta, \gamma}}^{(k)}$ to find the optimal steering angles and LED directivity index. However, searching all possible angles is unnecessary in many cases. For example, if all users are at one side of the room, we can narrow down the search space to that side of the room only, and decrease the computing complexity of the problem. In order to propose a method for narrowing down the search space, we provide the following propositions. In most use case scenarios, the transmitter LED is at a higher location than all of the users, and users are at a similar height. For these propositions, we assume users' heights are the same and therefore they all lie on the same plane. 

\begin{proposition} \label{prop1}
When there are two users in the system, optimal steering angle points to a location on the line segment between the location of the two users. \end{proposition}
\begin{proof}
See Appendix~\ref{Appendix1}.
\end{proof}
\begin{proposition} \label{prop2}
When there are more than two users that are on the same plane, optimal steering angle points to a location in the convex hull of the user locations, which is the smallest convex polygon that includes all locations.
\end{proposition}
\begin{proof}
See Appendix~\ref{Appendix2}.
\end{proof}

\begin{algorithm}[tb]
	\caption{The Graham Scan algorithm\cite{graham1972efficient}.}
	\begin{algorithmic}[1]
		\STATE Find the point with lowest $y$ value. If there are two points with the same y value, then choose the point with smaller x coordinate value. Make a list of points, and make this point the first one ($P[0]$).
		\STATE Sort the remaining $k-1$ points by the polar angle in counterclockwise order around $P[0]$, and add them to the list. If polar angles to two points are the same, delete the nearest point.
		\STATE For each point, if going to that point from the previous one takes a left turn keep that point in place. If it takes a right turn, remove previous points from the list until going to that point becomes a left turn.
		\STATE In the end, remaining points define the convex hull.
	\end{algorithmic}
	\label{graham_scan}
\end{algorithm}

According to the Proposition~\ref{prop1} and Proposition~\ref{prop2}, the optimal steering angle always points to some location within the convex hull or line segment that includes all user locations. To decrease the search space, we propose the following solution. We calculate the convex hull of user locations using Graham scan \cite{graham1972efficient}, which finds the convex hull of a finite set of points on the same plane. Let us assume we have $k$ points, where $k \ge 3$, and their Cartesian coordinates are $(x_k, y_k)$. The implementation of Graham scan for these points is explained in Algorithm~\ref{graham_scan}. After finding the convex hull, we search within $\alpha$ and $\beta$ angles that point to the hull. If there are more than three users that are on the same line segment, then the Graham scan returns two points, which results in a line segment instead of a hull. In this case, or in case there are only two users, the search space should be the $\alpha$ and $\beta$ angles that point to the line segment in between. 

\begin{wrapfigure}{r}{0.5\textwidth}
	\centering\vspace{-2mm}
	\includegraphics[width = 3 in]{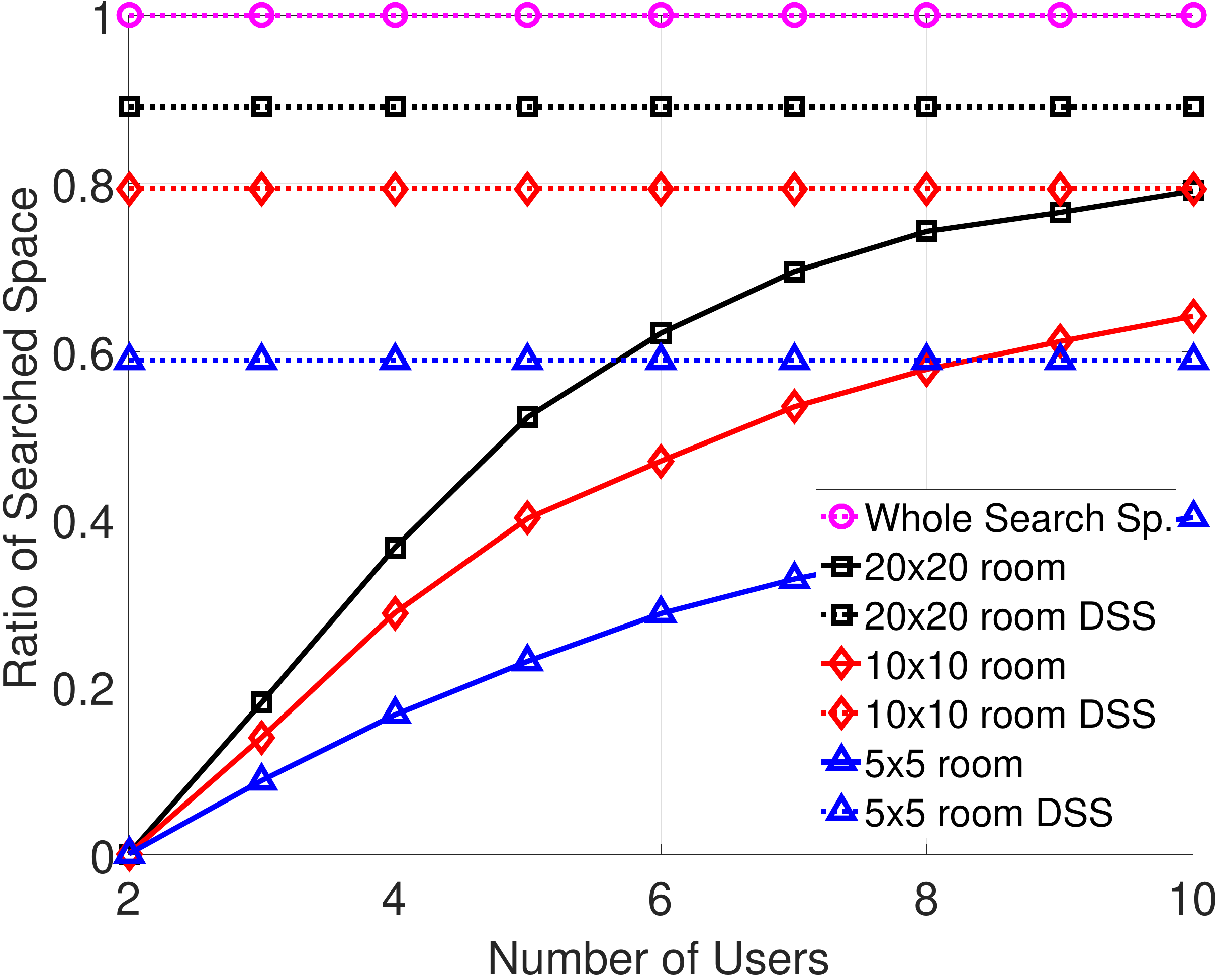}\vspace{-3mm}
	\caption{The ratio of the search space that needs to be scanned compared to the whole search space.}\vspace{-6mm}
	\label{SearchSpace}
\end{wrapfigure}

In Fig.~\ref{SearchSpace}, the ratio of the decreased search space (DSS) compared to the whole search space is shown. The simulation is done for an LED installed in the center of a room, at 4 m height, looking downwards. Users are located at $(x,y, 0.85)$ m, where $x$ and $y$ are uniformly distributed in the room. Three room sizes are considered: $5\times5$ m, $10\times10$ m, and $20\times20$ m. The elevation angle limits are $\alpha_{\rm min}=200$ and $\alpha_{\rm min}=340$. In Fig.~\ref{SearchSpace}, the dotted line with circle marker shows the ratio of the whole search space to itself, which is equal to one. The black dotted line with square markers shows the search space for the whole $20\times20$ m room, and the black solid line with square markers shows the DSS for the same room. For a low number of users, the search space is decreased significantly compared to the whole room, which means the algorithm provides a large gain in the computing time. When the number of users increases, they spread to a larger area and required search space increases too. Even when there are 10 users in the room, the algorithm reduces the search space about 10\%. The result is similar for smaller rooms, but they require a smaller search space compared to $20\times20$ m room. Overall, the proposed solution reduces the search space to 90\% - 1\% of the whole room depending on the number of users or the room size. 

\vspace{-3mm}

\section{Multiple Steerable Beams} \label{MultipleBeam}
In this section, we extend the solution in the previous section to multiple independently steerable beams case.

\vspace{-3mm}

\subsection{Steering Parameters for Multiple Beams}
In this subsection, we consider a transmitter that can steer multiple beams independently and therefore can track multiple users. When the number of users is equal or lower than the number of steerable beams, each user can be allocated a separate beam\footnote{In this paper, we do not address the problem of multiple beams serving to the same user.}. In case number of users are larger than the number of beams, users can be separated to clusters, and each cluster can be served with a separate beam as illustrated in Fig.~\ref{SteeringScenarios}(b). In order to cluster users, we introduce the VLC user clustering (VUC) algorithm, which is a modified $k$-means clustering technique. Each cluster corresponds to a separate beam. The VUC algorithm assigns users to the clusters based on the signal strength received from each beam, and finds the steering parameters for each beam, as described next in more detail. 

\subsubsection*{VLC User Clustering Algorithm}

Let there be $N$ steerable beams, and the steering angles and the directivity index of the $n$-th beam are $\alpha^{(n)}$, $\beta^{(n)}$, and $\gamma^{(n)}$, respectively. To initiate the algorithm, we randomly assign a single user to each cluster (i.e., assign first $N$ users to one cluster each). Initially, there are some unassigned users, but all users will be assigned to a cluster after the algorithm is completed. We have a total of $N$ clusters, and we repeat the following steps iteratively to find conclusive clusters and cluster centers. In the first step, we calculate the optimal steering parameters for the $n$-th beam, which are $\alpha^{(n)}$, $\beta^{(n)}$, and $\gamma^{(n)}$, solving the optimization problem in \eqref{betaOpt} as described in Section~\ref{PropSol}.A, for the users in the $n$-th cluster. We repeat this and find the steering parameters for each beam. While solving \eqref{betaOpt}, the search space should be decreased as described in Section~\ref{PropSol}.B to reduce the computation time. 

In the second step, we assign each user to the cluster whose beam provides the maximum signal strength to the user. We repeat these two steps until the steering parameters stay the same for two consecutive iterations. The proposed VUC algorithm is summarized in Algorithm~\ref{alg}, where $J(n)$ represent the set of users assigned to the $n$-th cluster, and $h_{k,n}$ denotes the channel gain between the $n$-th beam and the $k$-th user.

\begin{algorithm}[tb]
	\caption{The proposed VUC algorithm.}
	\begin{algorithmic}[1]
		\STATE Initialize: Assign user $n \rightarrow J(n) $ for $n = 1,...,N$
        \REPEAT 
        \FOR{$n$ = 1 to $N$} 
        \STATE {Solve \eqref{betaOpt} for the $n$-th beam and users in $J(n)$ to find the steering parameters of $n$-th beam ($\alpha^{(n)}$, $\beta^{(n)}$, and $\gamma^{(n)}$).} 
        \ENDFOR
        \FOR{$k$ = 1 to $K$} 
        \STATE { Find $n$ maximizing $h_{k,n}$, then assign user $k \rightarrow J(n)$.} 
        \ENDFOR
        \UNTIL{Steering parameters stay the same for two consecutive iterations.}
	\end{algorithmic}
	\label{alg}
\end{algorithm}

\vspace{-3mm}

\subsection{Power Optimization of Beams}
In this subsection, we discuss the optimal power allocation to different beams in order to maximize the sum rate of all users. The VUC algorithm works for a given transmit power of each beam and does not optimize the transmit powers. It is because the VUC algorithm aims at efficiently clustering users, and finding optimal steering parameters for each cluster. Considering that each cluster can have a different number of users, or users can have different received signal strength, we can improve the overall rate capacity by assigning different transmit powers to each beam. In a scenario with multiple beams serving different user clusters, each beam causes interference to users in other clusters. In this case, the SINR of the $k$-th user in the $n$-th cluster is given as
\begin{equation}
\xi_{k,n} = \frac{(p_nh_{k,n})^2}{N_0B + \displaystyle\sum_{\substack{m = 1 \\ m \neq n} }^{N}(p_mh_{k,m})^2 },
\label{eq14}
\end{equation}
where $p_n$ is the transmit power allocated to the $n$-th beam. The rate capacity of this user is
\begin{equation}
R_{k,n} = B\log(1+ \xi_{k,n} ).
\end{equation}
Then, the power optimization problem can be formulated as
\begin{equation}
\begin{split}
\tilde{\textbf{p}} = {\rm arg}\underset{ \textbf{p} }{\rm ~max}&~~  \displaystyle\sum_{k=1}^{K} \log(\tau_k R_{k,n}),\\
{\rm s.t.}~~ c_1:\quad & 0 \leq p_n \quad  \forall n \,,
\\ 
c_2:\quad & \sum_{n=1}^{N}p_n \leq p_{\rm max} \,,
\label{powerOpt}
\end{split} 
\end{equation}
where $\textbf{p}$ is the power allocation vector including power allocation of all beams. The constraint $c_1$ makes sure all power coefficients are positive, and the constraint $c_2$ makes sure their sum does not exceed the limit $p_{\rm max}$. The $\tau_k$ is the time allocation of the $k$-th user as addressed in \eqref{tauOpt}, and it is equal to $1/K_n$ where $K_n$ is the number of users served by the $n$-th beam. Note that using $\log$ at objective function is optional in this case. Even when we do not use it, more than one LED can be allocated some power level to maximize the overall sum rate. However, using logarithm can be preferred for fairness. 

The problem in \eqref{powerOpt} is non-convex because of the objective function. There are optimization parameters both on the numerator and the denominator of \eqref{eq14}, and linearizing the sum of logarithm of non-convex functions is not possible. To avoid this structure, we introduce auxiliary variables $\zeta_{k,n}$ and $\eta_{k,n}$ such that~\cite{he2017codebook}
\begin{equation}
\begin{split}
B\log(1+\zeta_{k,n}) \ge \eta_{k,n} \quad \forall \, k \,, \quad \mbox{and} \quad
\xi_{k,n} \ge \zeta_{k,n} \quad \forall \, k \, ,\label{auxil}
\end{split}
\end{equation}
where $\eta_{k,n}$ is a lower bound for the rate of $k$-th user in the $n$-th cluster, and $\zeta_{k,n}$ is a lower bound for the SINR of that user.

Using \eqref{auxil}, the problem in \eqref{powerOpt} becomes
\begin{equation}
\begin{split}
\tilde{\textbf{p}} = {\rm arg}\underset{ \textbf{p} }{\rm ~max}&~~  \displaystyle\sum_{k=1}^{K} \log(\tau_k \eta_{k,n}),\\
{\rm s.t.}~~ c_1:\quad & 0 \leq p_n \quad  \forall n \,,
\\ 
c_2:\quad & \sum_{n=1}^{N}p_n \leq p_{\rm max} \,, \\
c_3:\quad & B\log(1+\zeta_{k,n}) \ge \eta_{k,n} \quad \forall \, k \, , \\
c_4:\quad & \xi_{k,n} \ge \zeta_{k,n} \quad \forall \, k \, ,
\label{powerOpt2}
\end{split} 
\end{equation}
where $c_3$ and $c_4$ are added to satisfy \eqref{auxil}.
While the objective function in \eqref{powerOpt2} is now convex, the constraint $c_4$ is still non-convex, and the SINR expression is still there. To address this, we introduce another auxiliary variable set $\kappa_{k,n}$ which is an upper bound for the denominator of the $\xi_{k,n}$ given in \eqref{eq14}. Now we can replace $c_4$ with $c_5$ and $c_6$ which are given as:
\begin{equation}
\begin{split}
c_5:\quad & \frac{(h_{k,n}p_n)^2}{\kappa_{k,n}} \ge \zeta_{k,n} \quad \forall ~ k, \\
c_6:\quad & N_0B + \displaystyle\sum_{\substack{m = 1 \\ m \neq n} }^{N}(p_mh_{k,m})^2 \le \kappa_{k,n} \quad \forall ~ k.
\end{split} 
\end{equation}
Finally, the problem becomes
\begin{equation}
\begin{split}
\tilde{\textbf{p}} = {\rm arg}\underset{ \textbf{p}}{\rm ~max}&~~  \displaystyle\sum_{k=1}^{K} \log(\tau_k \eta_{k,n}),\\
{\rm s.t.}~~ c_1,\, c_2,\, c_3,\,& c_5,\, \mbox{ and } c_6 .
\label{powerOpt3}
\end{split} 
\end{equation}

The constraint $c_5$ in \eqref{powerOpt3} is still non-convex because of the expression $\frac{p_n^2}{\kappa_{k,n}}$, but it is in a simpler form and hence we can use MM procedure on this constraint. In order to use MM, we approximate the expression $\frac{p_n^2}{\kappa_{k,n}}$ for $k$-th user with multivariate first order Taylor series. This is a function of variables $p_n$ and $\kappa_{k,n}$, therefore we can express it as 
\begin{align}
\frac{p_n^2}{\kappa_{k,n}} = f(p_n, \kappa_{k,n}).
\label{taylor1}
\end{align}
The first order Taylor approximation for this function at point $p_n = a_n$ and $\kappa_{k,n} = b_{k,n}$ is:
\begin{align}
f(p_n, \kappa_{k,n}) \approx& f(a_n, b_{k,n}) + \frac{\partial f}{\partial p_n}(a_n, b_{k,n})(p_n-a_n) + \frac{\partial f}{\partial \kappa_{k,n}}(a_n, b_{k,n})(\kappa_{k,n} - b_{k,n}) \label{taylor2}\\
=& \frac{a_n^2}{b_{k,n}} + \frac{2a_n}{b_{k,n}}(p_n - a_n)
- \frac{a_n^2}{b_{k,n}^2}(\kappa_{k,n} - b_{k,n})
~=~ 2\frac{a_n}{b_{k,n}}p_n - \left(\frac{a_n}{b_{k,n}}\right)^2\kappa_{k,n}.
\label{taylor3}
\end{align}
The function in \eqref{taylor1} can be approximated by \eqref{taylor3} which is a convex expression. Inserting this expression into $c_5$, the constraint becomes
\begin{align}
c_5: \quad \left(h_{k,n}\right)^2\left(2\frac{a_n}{b_{k,n}}p_n - \left(\frac{a_n}{b_{k,n}}\right)^2\kappa_{k,n}\right) \ge \zeta_{k,n} \quad \forall ~ k \,. \nonumber
\end{align}

The MM procedure on \eqref{powerOpt3} operates iteratively. We first solve the problem for some initial values of $a_n$ and $b_{k,n}$. We do not need to carry out the iterations in two dimensions, because the constraint is only dependent on the division of ${a_n}$ and ${b_{k,n}}$. Therefore we update the value of $\frac{a_n}{b_{k,n}}$ at each iteration until it stays the same for two consecutive iterations, or the change between two consecutive iterations is not appreciable.

\vspace{-3mm}

\section{VLC NOMA for Beam Steering and User Clustering} \label{NOMAsection}
In the previous section, we assumed that all users in a cluster are served by TDMA. An additional approach to improve user rates is to implement NOMA for some of the users. Since the LEDs are directional, not all users in a cluster receive similar signal strength. Even though the optimization problem in \eqref{betaOpt} considers the fairness among users, due to the Lambertian pattern of the signal, we expect some users to receive much higher signal strength compared to others. In this case, an opportunistic approach would be to employ NOMA to exploit this uneven signal strength distribution and improve the overall data rate of users. 

In this section, we consider NOMA application for two users that are served by $n$-th LED, by coupling users whose channel gains are distinctively different. We will denote these two users as user-1 and user-2, where user-1 has the weaker channel gain ($h_{1,n}<<h_{2,n}$). The general VLC NOMA signal model is provided in Section~\ref{SystemModel}.D. The achievable data rate for these two users are given as follows:
\begin{align}
R_{i,n} = \log_2(1+\xi_{i})
\end{align}
where the SINRs for user-1 and user-2 are given as
\begin{align}
\xi_{1} = \frac{(h_{1,n}\rho_1p_n)^2}{ N_0B + \sum\limits_{\substack{m = 1 \\ m \neq n} }^{N} (h_{1,m}p_m)^2 + (h_{1,n}\rho_2p_n)^2 } \,, \quad \quad
\xi_{2} = \frac{(h_{2,n}\rho_2p_n)^2}{ N_0B + \sum\limits_{\substack{m = 1 \\ m \neq n} }^{N} (h_{2,m}p_m)^2 } \,.
\label{NOMA_SINRs}
\end{align}
The first interference term of both SINR expressions come from the interference caused by other beams. The user-1 has another interference component, which is the signal message of user-2. The user-2, on the other hand, does not have this interference since it detects and cancels the message of user-1. This rate is conditioned on the event that user-2 successfully detects and cancels the signal of user-1. Let $\xi_{2\rightarrow1}$ denote the SINR for user-2 to detect the message for user-1, and $\xi_1^{\rm *}$ as the targeted SINR for successful message detection at user-1. Then the condition can be expressed as $\xi_{2\rightarrow1} \geq \xi_1^{*}$, or explicitly
\begin{align}
\frac{( h_{2,n}\rho_1p_n )^2}{N_0B + \sum\limits_{\substack{m = 1 \\ m \neq n} }^{N} (h_{2,m}p_m)^2 + ( h_{2,n}\rho_2p_n )^2} \geq \xi_1^{*}.
\label{SINR_Thre}
\end{align}

\subsection{NOMA Parameter Optimization Problem}
In order to maximize the user rate, the parameters to be optimized include the power allocation of each beam and NOMA power coefficients of NOMA users. However, such an optimization problem would be too complex to solve. Firstly, in such a problem, deciding which user pair will utilize NOMA is difficult because the power allocation of each beam is unknown. An iterative solution can be proposed where the solution updates the power allocation, NOMA user pairs, and NOMA power coefficients of these users at each iteration. However, this solution would show erratic behavior since the selected NOMA pairs would introduce a non-continuous objective function because user pair selection is a binary optimization problem. Due to the complexity of this problem, we do not propose a single step solution. 

Instead of solving the problem in a single step, we can use the power optimization that is proposed in Section~\ref{MultipleBeam} as the first step of the solution, choose NOMA user pairs, and optimize the NOMA coefficients of these users as the second step of the problem. This solution is guaranteed to improve the overall sum rate as long as the sum rate of NOMA users is improved because the NOMA coefficients of a user pair do not affect the interference received by other users. The same statement is also valid for logarithmic sum rate. The sum of the logarithm of the user rates is proportional to the multiplication of the user rates. If the multiplication of the rates of two users increases, overall multiplication of the user rates increases too. 

In order to implement the proposed solution, we need to decide the NOMA user pairs. It is well known that NOMA is more efficient when the channel gains are more distinct \cite{7273963}. The simplest solution is to pair the users with the highest and lowest channel gains in each beam \cite{7273963, haas_noma} if they meet the SINR threshold criteria. After that, remaining users with the most distinctive channel gains can be paired if they meet the same criteria. For any user pair, the optimal NOMA coefficients can be found by solving the following problem:
\begin{equation}
\begin{split}
\tilde{\boldsymbol{\rho}} = {\rm arg}\underset{ \boldsymbol{\rho} }{\rm ~max}&~~  \displaystyle\sum_{i=1}^{2} \log( R_{i,n}),\\
{\rm s.t.}~~ c_1:\quad & \rho_1^2 + \rho_2^2 = 1 \,,
\\ 
c_2:\quad & \xi_{2\rightarrow1} \geq \xi_1^{*} \,,
\label{NOMA_Opt}
\end{split}
\end{equation}
where the constraint $c_1$ is for the preservation of energy, and the constraint $c_2$ is to ensure successful interference cancellation at user-2. The logarithm at the objective function is optional. 

\vspace{-3mm}

\subsection{Proposed Solution for NOMA Parameter Optimization}
In \eqref{NOMA_Opt}, the objective function and both constraints are non-convex. In the objective function, the only non-convex expression is the $\xi_1$, because there are optimization parameters both in the nominator and the denominator\cite{he2017codebook}. In order to handle this expression, as we did in the solution of \eqref{powerOpt}, we introduce an auxiliary parameter $\zeta$ as a lower bound of SINR of user-1: $\xi_{1} \ge \zeta  \,$ . Now the problem in \eqref{NOMA_Opt} becomes
\begin{equation}
\begin{split}
\tilde{\boldsymbol{\rho}} = {\rm arg}\underset{ \boldsymbol{\rho} }{\rm ~max}&~  \log( \log (1+\zeta)) + \log(\log(1+\xi_2)),
\label{NOMA_Opt2}
\end{split}
\end{equation}
subject to $c_1$ and $c_2$ in \eqref{NOMA_Opt}, and $c_3:~ \xi_{1} \ge \zeta$, to satisfy the lower bound. Note that the $\xi_{1}$ is replaced with $\zeta$ in the objective function. In order to handle the constraint $c_1$, we introduce the another auxiliary variable $\eta$ such that $\eta = \rho_1^2$, and we replace the all $\rho_1^2$s with $\eta$, all $\rho_2^2$s with $1-\eta$ in \eqref{NOMA_SINRs} and \eqref{SINR_Thre}. We also replace the constraint $c_1$ as
$c_1: 0 \le \eta \le 1$, 
to make sure coefficients stay within the limit $[0,1]$. 

In \eqref{NOMA_Opt2}, the constraints $c_2$ and $c_3$ are non-convex. For these two expressions, we introduce two more auxiliary variables $\kappa_1$ and $\kappa_2$, and replace the constraint $c_2$ with the following expressions: 
\begin{align}
c_{2.1}:\quad& \frac{( h_{2,n}p_n )^2\eta}{\kappa_1} \geq \xi_1^{*} \,, \label{c_2NOMA} \\
c_{2.2}:\quad& \kappa_1 \geq N_0B + \sum\limits_{\substack{m = 1 \\ m \neq n} }^{N} (h_{2,m}p_m)^2 + ( h_{2,n}p_n )^2(1-\eta)  \,, \nonumber
\end{align}
while the constraint $c_3$ is replaced with:
\begin{align}
c_{3.1}:\quad& \frac{(h_{1,n}p_n)^2\eta}{ \kappa_2 } \geq \zeta \,, \label{c_3NOMA} \\
c_{3.2}:\quad& \kappa_2 \geq N_0B + \sum\limits_{\substack{m = 1 \\ m \neq n} }^{N} (h_{1,m}p_m)^2 + (h_{1,n}p_n)^2(1-\eta). \nonumber
\end{align}
The constraints $c_{2.2}$ and $c_{3.2}$ are convex because they are in the form of comparison of two optimization variables with some constant multipliers. The $\kappa_1$, $\kappa_2$, and $\eta$ are the optimization parameters in these constraints, while all other parameters are constants. The constraint $c_{2.1}$ can be expressed in the form of comparison of two optimization parameters by sending the $\kappa_1$ to the other side of the equation. On the other hand, the constraint $c_{3.1}$ is non-convex, because it includes the division of two optimization parameters, and another optimization parameter $\zeta$ on the other side of the inequality. In order to deal with that, we replace the expression $\frac{\eta}{\kappa_2}$ with its multivariate first order Taylor series expansion as we did with \eqref{taylor1}. The Taylor series expansion of the expression 
$\eta/\kappa_2 = f(\eta, \kappa_2)$, 
when evaluated at point $\eta = a$ and $\kappa_{2} = b$ is given as:
\begin{align}
f(\eta, \kappa_2) &\approx f(a, b)  
+ \frac{\partial f}{\partial \eta}(a, b)(\eta-a) + \frac{\partial f}{\partial \kappa_{2}}(a, b)(\kappa_{2} - b) \nonumber \\
&= \frac{a}{b} + \frac{1}{b}(\eta-a) - \frac{a}{b^2}(\kappa_{2} - b) ~=~ \frac{\eta}{b} - \frac{a}{b^2}\kappa_{2} - \frac{a}{b} \,. \label{TaylorNoma}
\end{align}
Now we can insert the Taylor series expansions to constraint $c_{3.1}$, and the optimization problem is finally convex.

In order to solve the problem, we need to implement MM procedure over parameters $a$ and $b$. We start with some initial values of these parameters, solve the convex optimization problem, update the MM parameters for the found values of $\eta$ and $\kappa_2$, and repeat until the parameters stay unchanged for two consecutive iterations. With the suggested changes, the final optimization problem becomes the following:
\begin{align}
\tilde{\boldsymbol{\eta}}, \tilde{\boldsymbol{\zeta}}, \tilde{\boldsymbol{\kappa_1}}, \tilde{\boldsymbol{\kappa_2}} &= {\rm arg}\!\!\underset{ \boldsymbol{\eta, \zeta, \kappa_1, \kappa_2} }{\rm ~max}~ \!\!\log (1+\zeta) + \log(1+\xi_2), \label{NOMA_Opt3}
\end{align}
subject to $c_1: 0 \le \eta \le 1$, $c_{2.1}$ and $c_{2.2}$ in \eqref{c_2NOMA}, $c_{3.2}$ in \eqref{c_3NOMA}, and 
\begin{align}
c_{3.1}:\quad& (h_{1,n}p_n)^2 \left( \frac{\eta}{b} - \frac{a}{b^2}\kappa_{2} - \frac{a}{b} \right) \geq \zeta \,, \nonumber
\end{align}
where $\eta$ and $\kappa_2$ are evaluated at $a$ and $b$. This problem can be directly solved with convex optimization tools such as CVX \cite{IRM_Cheth_Chandra}.

\vspace{-3mm}

\section{Simulation Results} \label{SimRes}
We conduct computer simulations using MATLAB, where we consider a square room of dimensions 8~m $\times$ 8~m $\times$ 4~m. The LED transmitter is located in the center at ceiling level, and receivers are distributed at uniformly random locations at 0.85~m height, facing upwards. A wide FOV angle is considered so that the LED is always within LOS of the receivers. The remaining simulation parameters are given in Table~\ref{SimPar}.

\begin{table}\vspace{-1mm}
	\caption{Simulation parameters.}
		\vspace{-7mm}
	\begin{center}
		\begin{tabular}{ |c|c| }
			\hline
			The transmit power, $p$ & 1~W \\
			\hline
			The receiver responsivity, $r$, and effective surface area, $A_{\rm r}$ & 1~A/W, and $1$~cm$^2$ \\
			\hline
			The modulation bandwidth $B$ &  20 MHz \\
			\hline
			The AWGN spectral density, $N_0$ & 2.5$\times 10^{-20}$ A$^2$/Hz \\
			\hline
			LED directivity indexes $\gamma_{\rm min}$, $\gamma_{\rm max}$, and $\gamma_{\rm def}$ & 1, 15, and 5 \\
			\hline
			Elevation angle limits $\alpha_{\rm min}$ and $\alpha_{\rm max}$ & $200^\circ$ and $340^\circ$ \\
			\hline
		\end{tabular}
		\label{SimPar}
	\end{center}
\vspace{-8mm}
\end{table}

\begin{figure}[tp]
	\centering
	\vspace{-1mm}
	\subfigure[Single steerable beam AP.]{
		\includegraphics[width = 3 in] {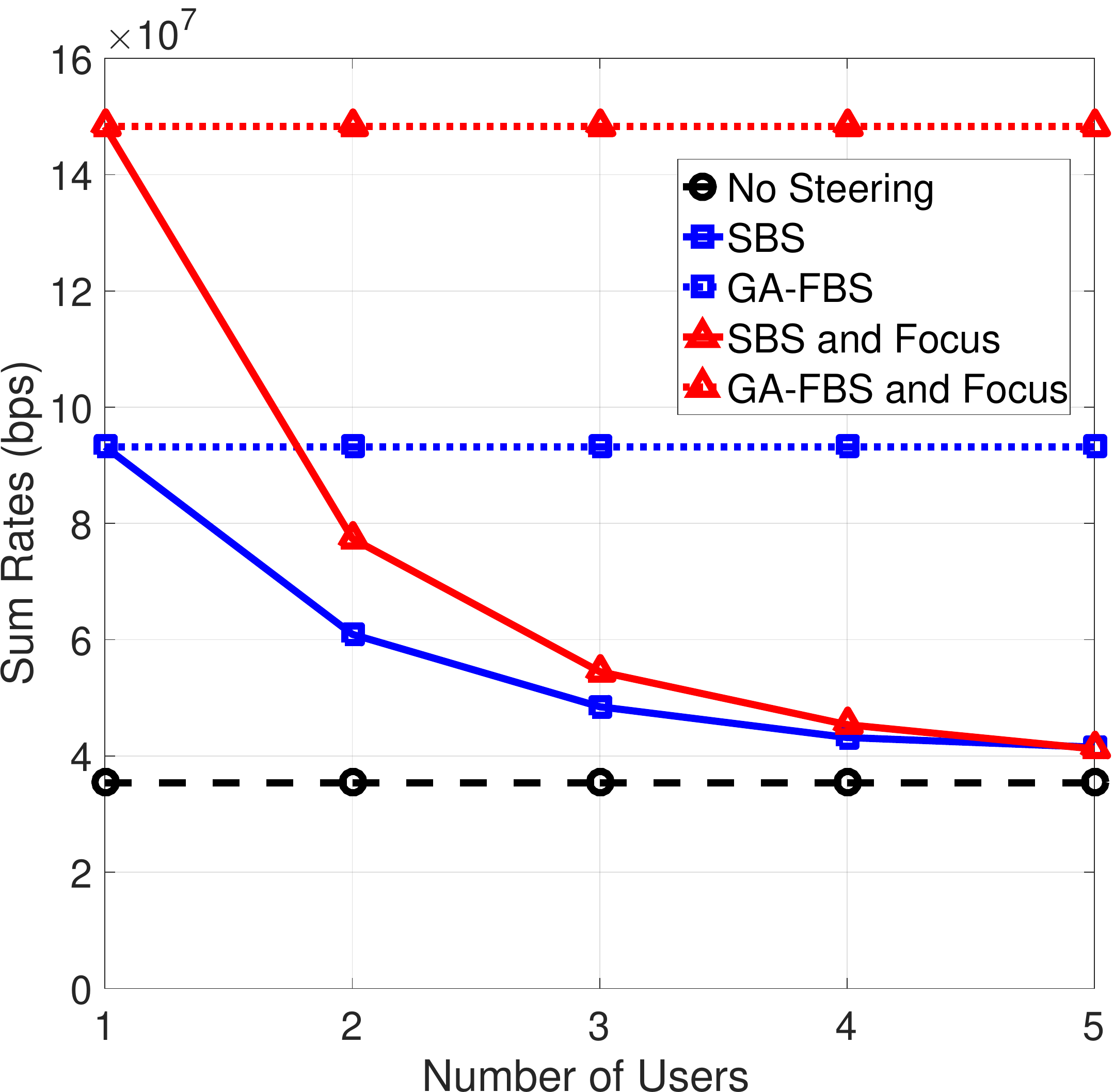}
		\label{RateSingleBeam}
	}~
	\subfigure[Three independently steerable beams.]{
		\includegraphics[width = 3 in] {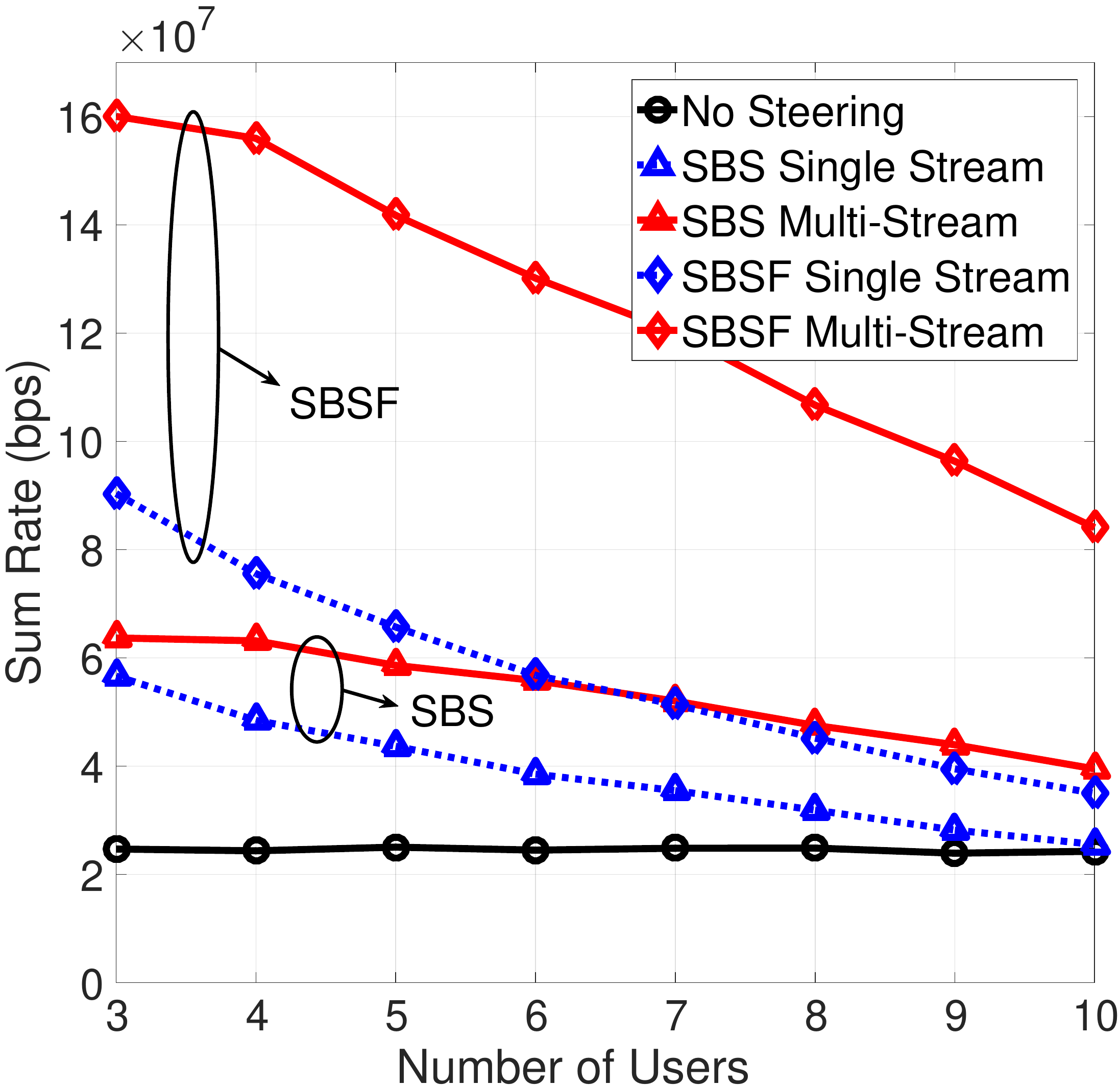}
		\label{RateClustering}
	}\vspace{-1mm}
    \caption{ The sum rate with single and multiple steerable beams. }\vspace{-5mm}
\end{figure}

\subsection{Single LED Beam Steering}

In Fig.~\ref{RateSingleBeam}, the sum rate of users are shown when the transmitter has a single steerable beam. We simulate three different scenarios. The first one is labeled as ``No Steering", where the LED beam is not steered and faced downwards with the default directivity index $\gamma_{\rm def}$. The second one is labeled as slow beam steering (SBS), where the beam is steered as described in Section~\ref{SystemModel}.A. In this scheme, we assume the directivity index cannot be changed, and equal to $\gamma_{\rm def}$. The third scenario is labeled as slow beam steering and focus (SBSF), where both beam orientation and directivity index are optimized. For comparison, we also consider a \emph{genie-aided} fast beam steering (GA-FBS) approach as an upper bound on the sum rate. In particular, while settling time for steering may be on the order of 5 ms in practice~\cite{Morrison:15}, we assume that we can instantaneously steer beams to each scheduled user so that the beam is completely steered towards the user that is being served at each time slot. 

Results in Fig.~\ref{RateSingleBeam} show that when there is a single user, a significant gain on the sum rate can be achieved with steering and focusing. In this case, optimal steering angles point to the direction of the user, and the optimal directivity index is high since the user is on the exact direction of the beam. When the number of users increases, the total rate achievable with steering decreases. The optimization maximizes the sum of the logarithm of rates to serve all users simultaneously; therefore the beam orientation does not point to a single user, and the optimal directivity index gets lower. Since users are not in the exact direction of the beam, the sum rate decreases as the number of users increases. The sum rate for GA-FBS schemes does not decrease, because the beam is steered towards the receiving user at each time interval, and we consider the average rate over a large number of user locations. 

\vspace{-3mm}

\subsection{Multiple LEDs and User Clustering}
In Fig.~\ref{RateClustering}, the sum rate of users are shown when the AP has three independently steerable beams. The transmit power of these beams are $p/3$ (versus $p$ that was used in Fig.~\ref{RateSingleBeam}) for a fair comparison. For this simulation, we consider two different multiple access schemes. The first one is labeled as ``single stream" and shown with dashed blue lines, where all beams transmit the same signal to avoid any interference. In this scheme, the signal strength is higher, and the interference is zero. However, all the users are served with time division of a single stream, therefore they are allocated a lower amount of TDMA time resources. In the multi-stream scheme shown with solid red lines, all beams transmit a different stream to the users assigned to them. Since Fig.~\ref{RateClustering} shows results for an AP with three independently steerable beams, the multi-stream scheme has three different streams. If multiple users are assigned to the same beam, they share the channel with TDMA. Due to the use of spatial diversity and higher time allocation to the users, this scheme may offer higher rates than the single stream scheme.

\begin{figure}[tp]
	\centering
	\vspace{-1mm}
	\subfigure[The CDF of individual user rates with three steerable beams and six users. The rates are shown in logarithmic scale.]{
		\includegraphics[width = 3 in] {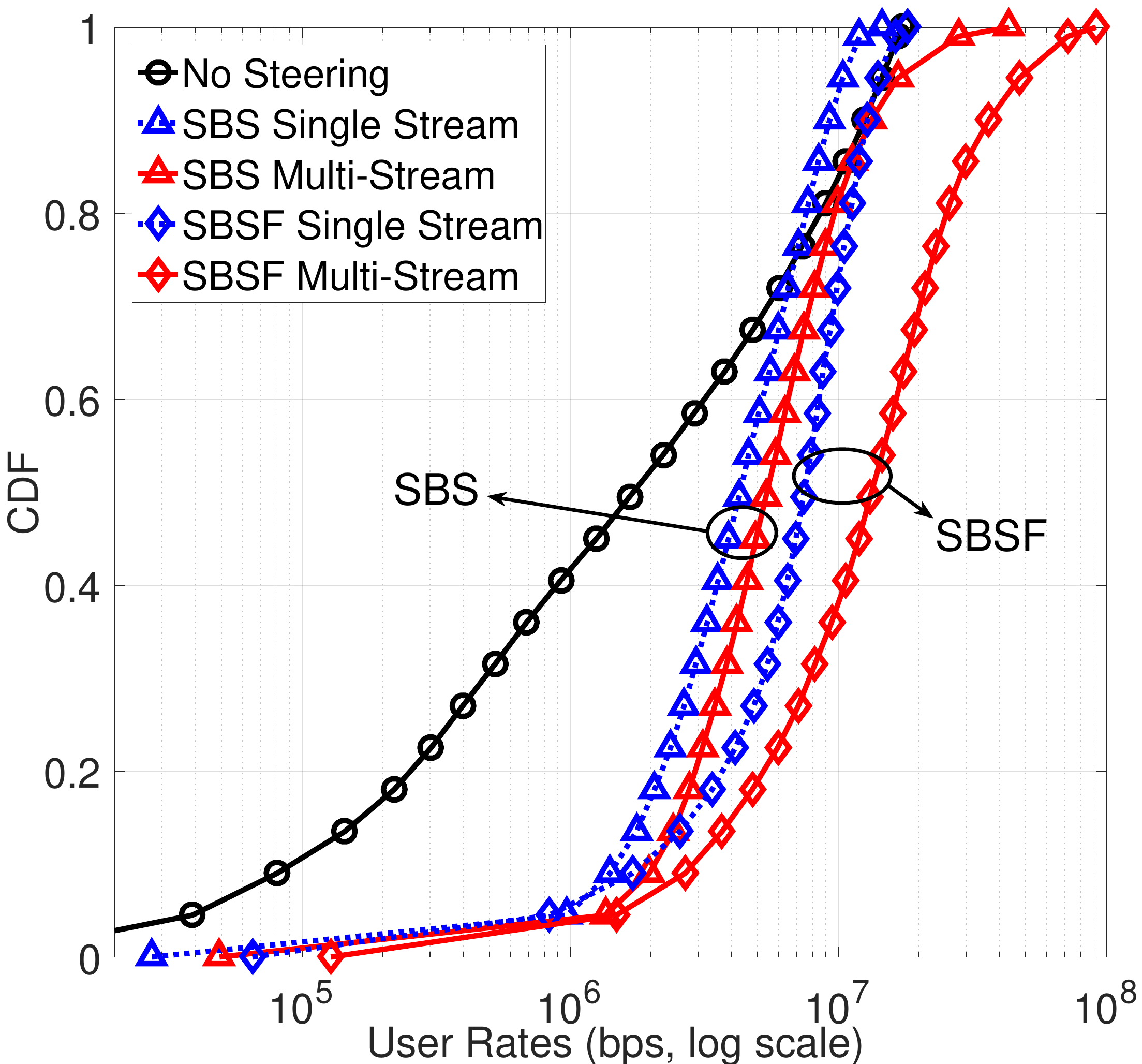}
		\label{CDF}
	}~~
	\subfigure[The sum rate when the AP serves 10 users and has varying number of steerable beams.]{
		\includegraphics[width = 3 in] {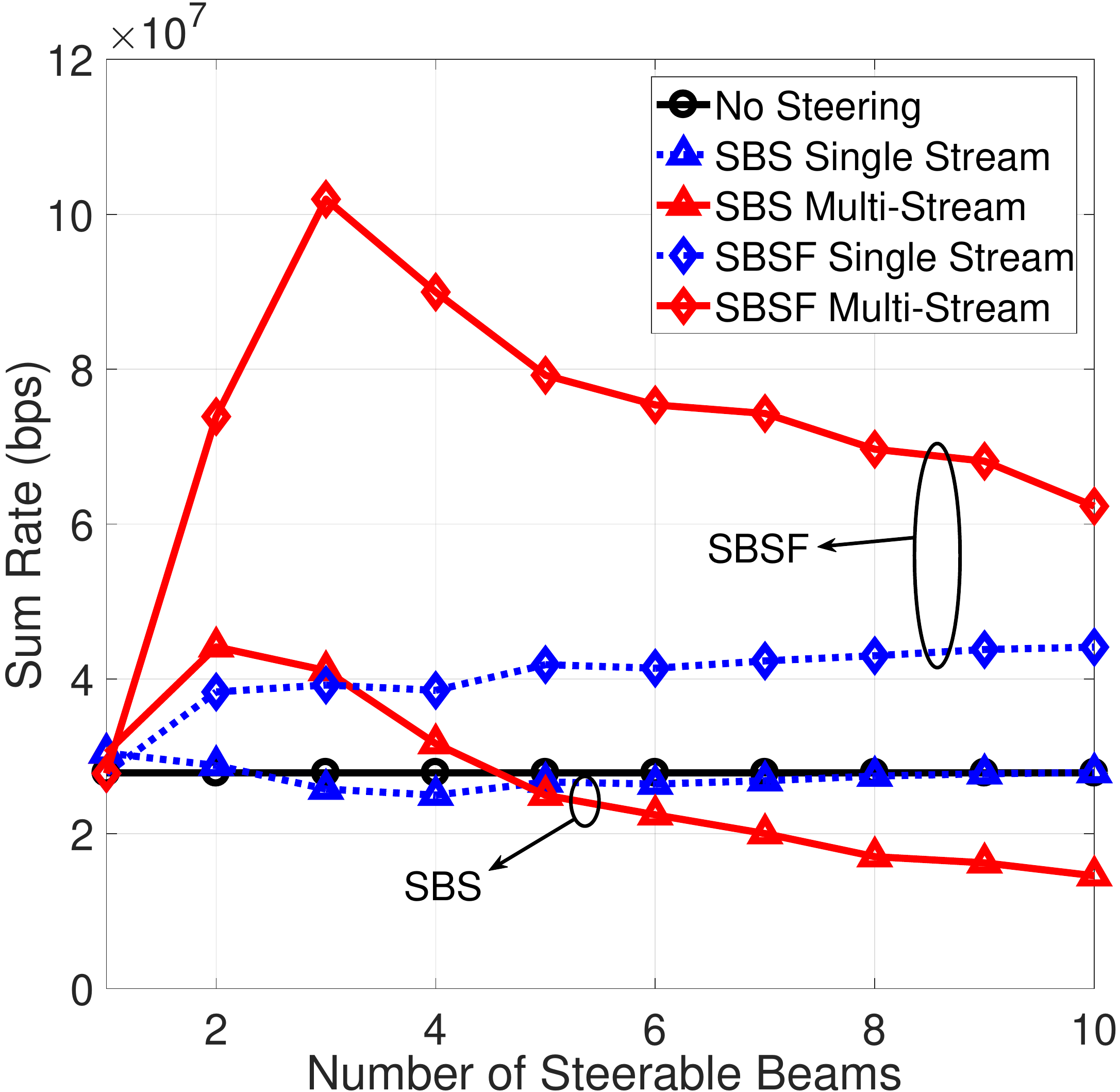}
		\label{VaryingBeams}
	}\vspace{-1mm}
    \caption{User rates and sum rates with multiple steerable beams. }\vspace{-5mm}
\end{figure}

As seen in Fig.~\ref{RateClustering}, the SBSF multi-stream provides the highest sum rates. The SBS multi-stream does not relatively perform well, especially with the lower number of users. In this scheme, some users suffer heavy interference because the directivity of the beams cannot be adjusted as needed. With the SBS single stream scheme, the sum rate decreases and approaches to no steering scheme with the increasing number of users. Since the ratio of users to the number of beams increases significantly, steering becomes less effective. Note that in Fig.~\ref{RateClustering} the sum rates do not decrease rapidly as in Fig.~\ref{RateSingleBeam}, especially sum rates of multi-stream schemes. This is due to VUC algorithm clustering users together that can receive high signal strength through a single beam. In Fig.~\ref{CDF}, the cumulative distribution function (CDF) of user rates are shown for six users and three steerable beams, as in the case of Fig.~\ref{RateClustering}. The steering provides more uniform distribution of user rates in comparison to no steering scheme since the optimization problem maximizes the sum of the logarithm of rates and provides a fairer resource allocation.

In Fig.~\ref{VaryingBeams}, the sum rates are shown for 10 users with a varying number of independently steerable beams. The transmit power of each beam is $p/N$, where $N$ is the number of beams. SBSF with multi-stream provides the highest sum rate, which is maximized at three beams per 10 users where the sum rate exceeds the four times of no steering scheme. The higher number of beams means better steering accuracy and higher received signal strength, however, it also causes higher interference in the multi-stream scheme and a lower transmit power per beam. The ideal user count per beam ratio may change based on the size of the room or the total number of users in the room. The SBS multi-stream scheme provides a lower data rate than no steering scheme if the number of steerable beams is high. This is because the signal strength of each user is low, and the interference from other beams is high. In no steering scheme and single stream schemes, there is no interference since all users are served in turn with TDMA. The SBS single stream scheme falls below no steering for 3 or 4 beams, which is possible because the optimization maximizes the sum of logarithmic rates instead of the sum rate. 

\begin{figure}[tp]
	\centering
	\vspace{-1mm}
	\subfigure[Varying number of users and the AP has three independently steerable beams. ]{
		\includegraphics[width = 3 in] {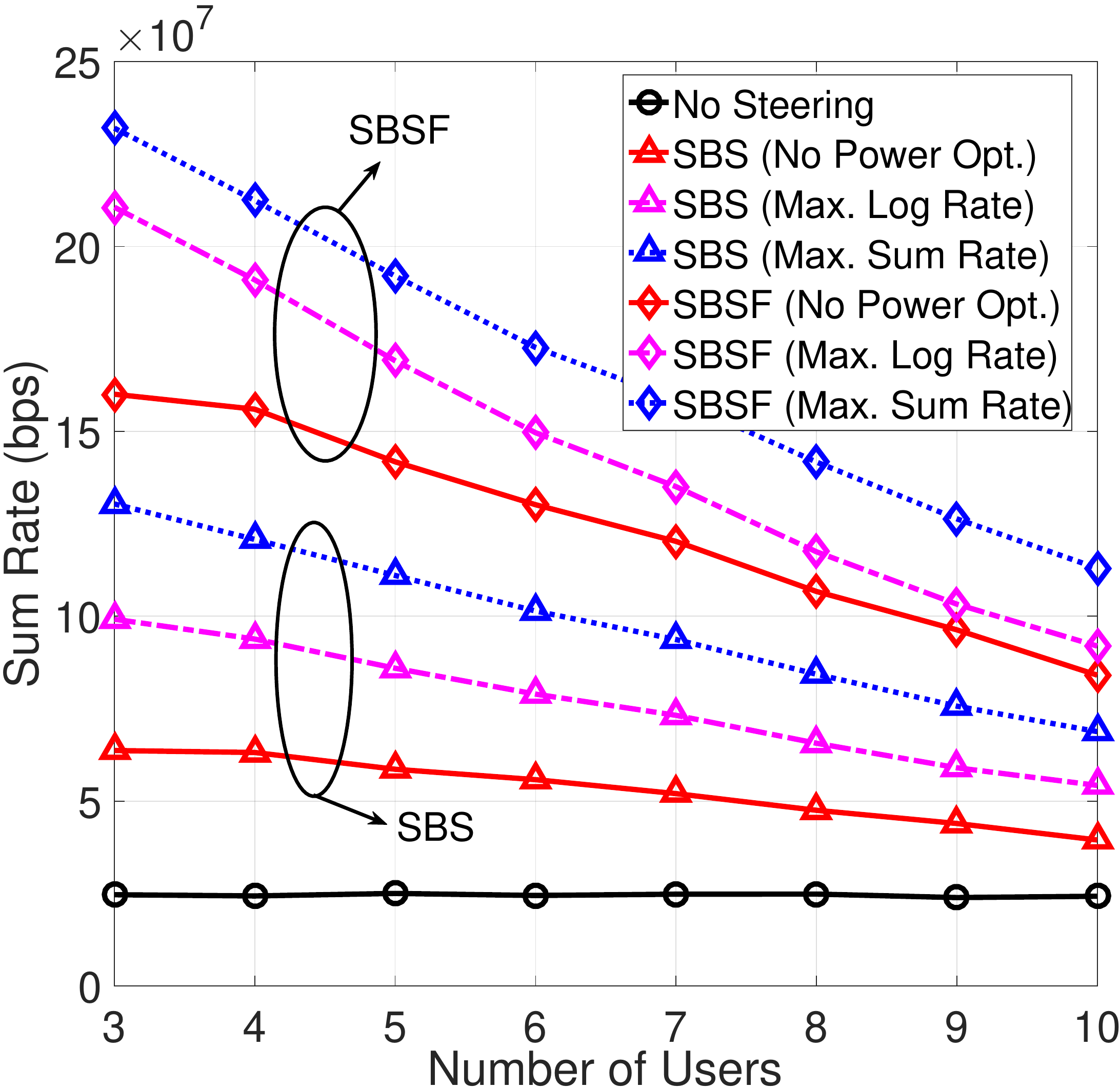}
		\label{RateClusteringOptimization}
	}~~
	\subfigure[Varying number of steerable beams with 10 users.]{
		\includegraphics[width = 3 in] {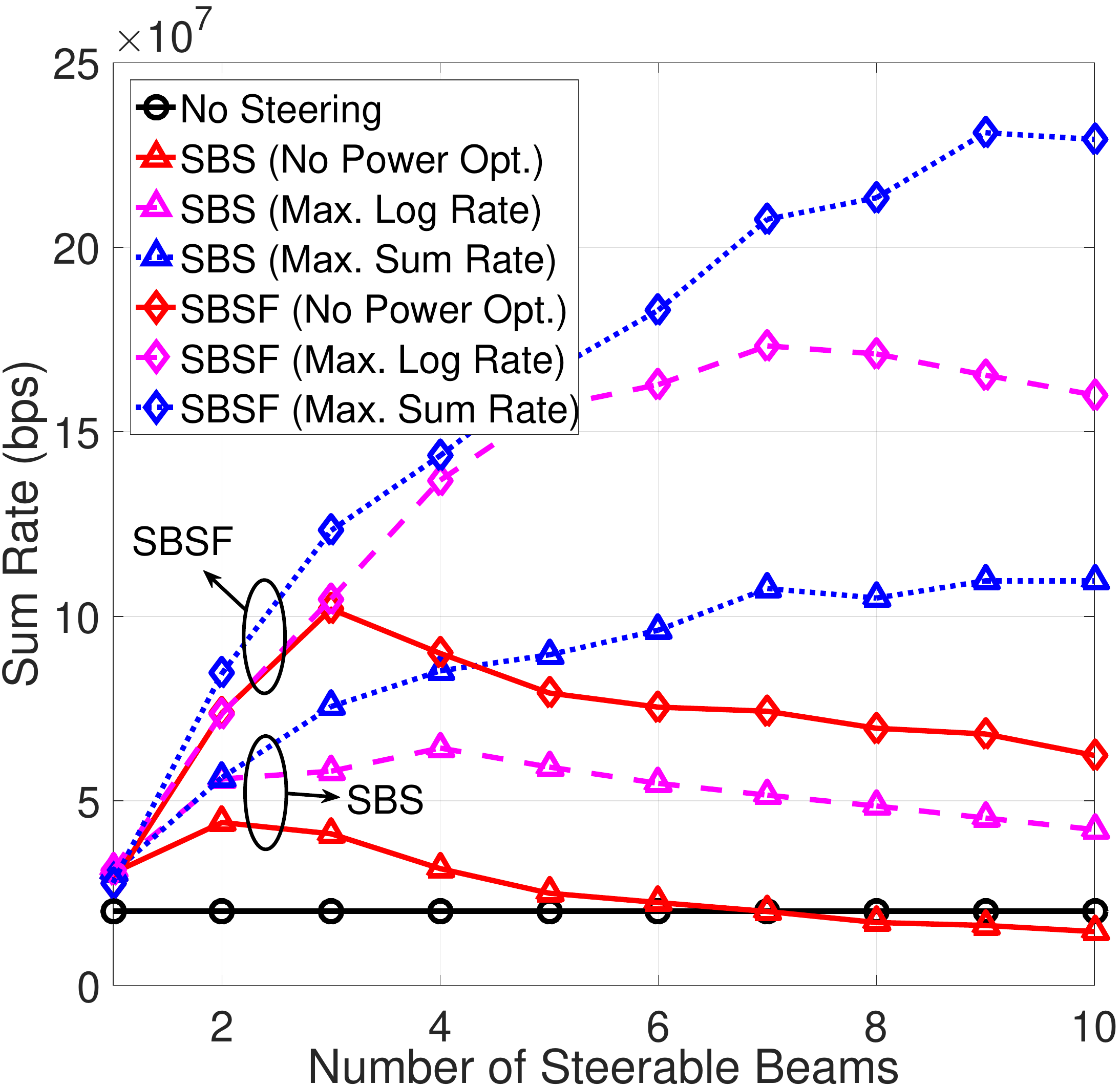}
		\label{VaryingBeamOptimization}
	}\vspace{-1mm}
    \caption{ The sum rate of users with power optimization where the beams transmit different streams (multi-stream). }\vspace{-5mm}
\end{figure}

\vspace{-3mm}

\subsection{Beam Power Optimization}
In Fig.~\ref{RateClusteringOptimization}, the sum rates of users are shown for an AP with 3 steerable multi-stream beams when the beam power optimization as in \eqref{powerOpt2} is used. The results for single-stream beams are not included due to poor performance and the requirement of a separate optimization solution. The results of power optimization for maximizing the sum rates are labeled as "Max. Sum Rate", and shown in dotted lines. The power optimization significantly increases the sum rate of the users, where the rate gain is between 30 - 70 Mbps for both SBS and SBSF. The rate gain is provided by assigning more power to the LEDs that have stronger LOS connection with users or serving more users overall. The sum rates of power optimization for maximizing the sum of the logarithm of the rates are labeled as "Max. Log Rate", and shown in dashed lines. In this case, there is a sum rate gain compared to "No Power Opt.", but the gain is not as high as the maximization of the sum rate. 

In Fig.~\ref{VaryingBeamOptimization}, the sum rates are shown for an AP with varying number of steerable beams. When there is no power optimization, the sum rate decreases for a high number of steerable beams. However, with the power optimization, the sum rate increases consistently. This is thanks to the interference adjustment feature of the power optimization solution. Since the power allocation is done considering the interference to other users, the higher number of beams can be utilized more efficiently. With 10 beams, the sum rate of the maximum sum rate case reaches to 10 times of the sum rate of no steering scheme.

\begin{wrapfigure}{r}{0.5\textwidth}
	\centering\vspace{-3mm}
	\includegraphics[width = 3.1 in]{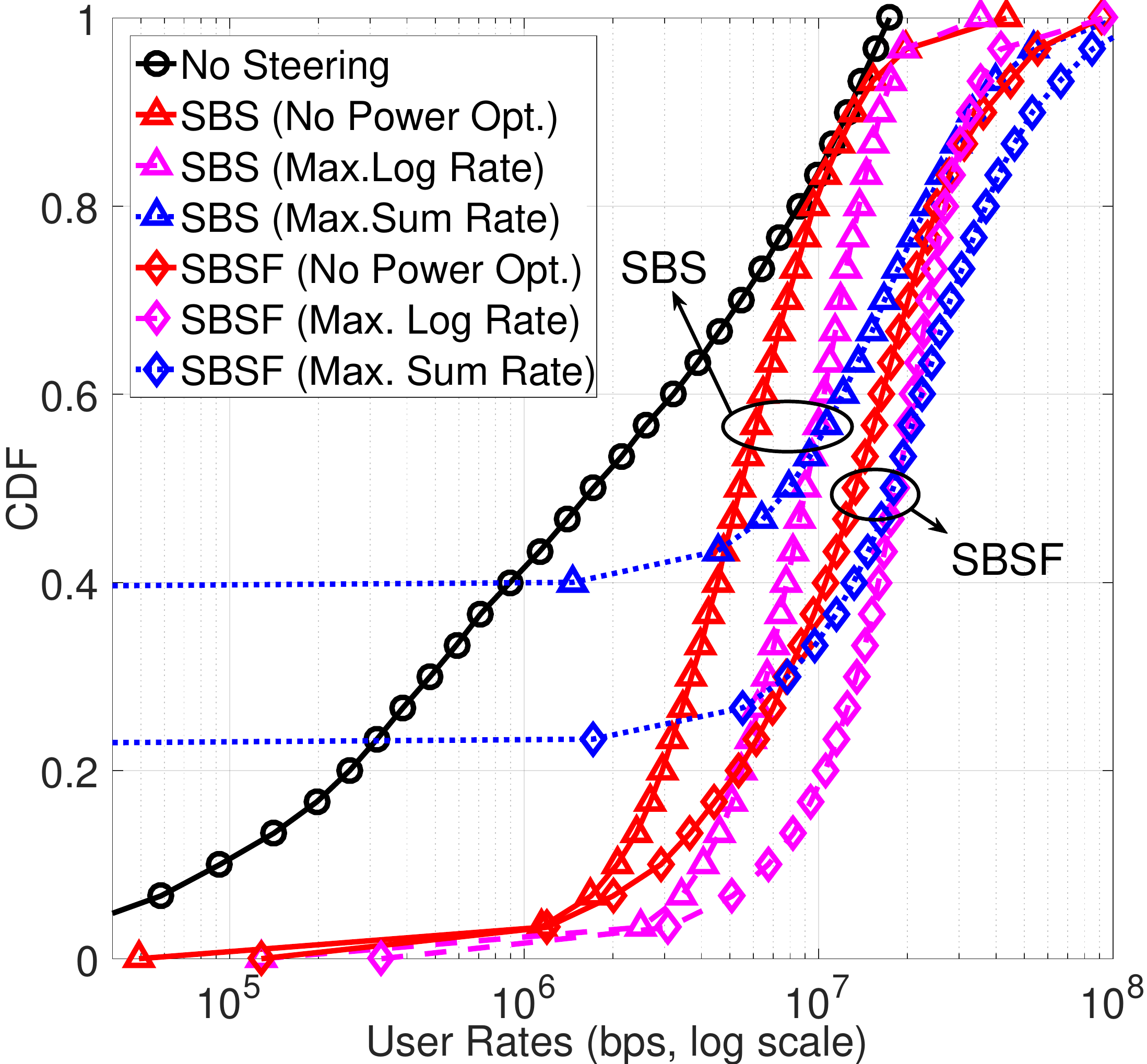}
	\vspace{-4mm}
	\caption{The CDF of individual user rates with three steerable beams and six users. The rates are shown in logarithmic scale.}
	\label{CDFPowOpt}\vspace{-5mm}
\end{wrapfigure}

In Fig.~\ref{CDFPowOpt} the CDF of individual user rates are shown for six user and three steerable beams case. The figure includes the data rates after power optimization of the beams for maximizing the sum rate and the logarithmic sum rate. Power optimization for maximum sum rate leaves some users without service but provides some other users much higher data rates. On the other hand, optimization for maximizing the logarithmic sum rate serves all users and increases the data rates overall. In this case, the low-rate users have more gain compared to high-rate users. It is also seen that users with the highest data rates actually lose some data rate after this optimization. 

\vspace{-3mm}

\subsection{NOMA}
In this subsection, simulations are conducted for three steerable beams setting and 10 users, with other parameters being the same as previous simulations. Users in the same cluster are paired to be served by NOMA. The users with most distinctive channel gains are paired if they meet the SINR threshold $\xi_1^* = 3$ as described in \eqref{SINR_Thre}. Then the remaining users are paired if they meet the same threshold. The users that are not paired are served by TDMA and get half the time allocation of NOMA user pairs. 

In Fig.~\ref{VaryingCoeff}, sum and individual data rates of two users using NOMA are shown for different small power coefficient ($\rho_2$) values. For comparison, the TDMA rates are also shown for the same users if they were not served by NOMA. The reason TDMA rates slightly increases with $\rho_2$ is that as $\rho_2$ increases fewer user pairs can achieve the threshold for NOMA, and their corresponding TDMA rates are slightly higher. The NOMA provides a gain in the sum rate compared to TDMA for all cases, except the case where $\rho_2$ is nearly equal to zero. When $\rho_2$ is between $0.1$ and $0.12$, both users have data rate gain compared to TDMA. As $\rho_2$ increases the sum NOMA rate slightly increases, however it causes an unfair allocation since the weak user's data rate decreases even further. The power coefficients should be selected by considering the trade-off between fairness and higher sum rate. 

\begin{figure}[tp]
	\centering
	\vspace{-1mm}
	\subfigure[User rates for varying $\rho_2$. Sum rate is the sum of the rates of weak and strong users.]{
		\includegraphics[width = 3 in] {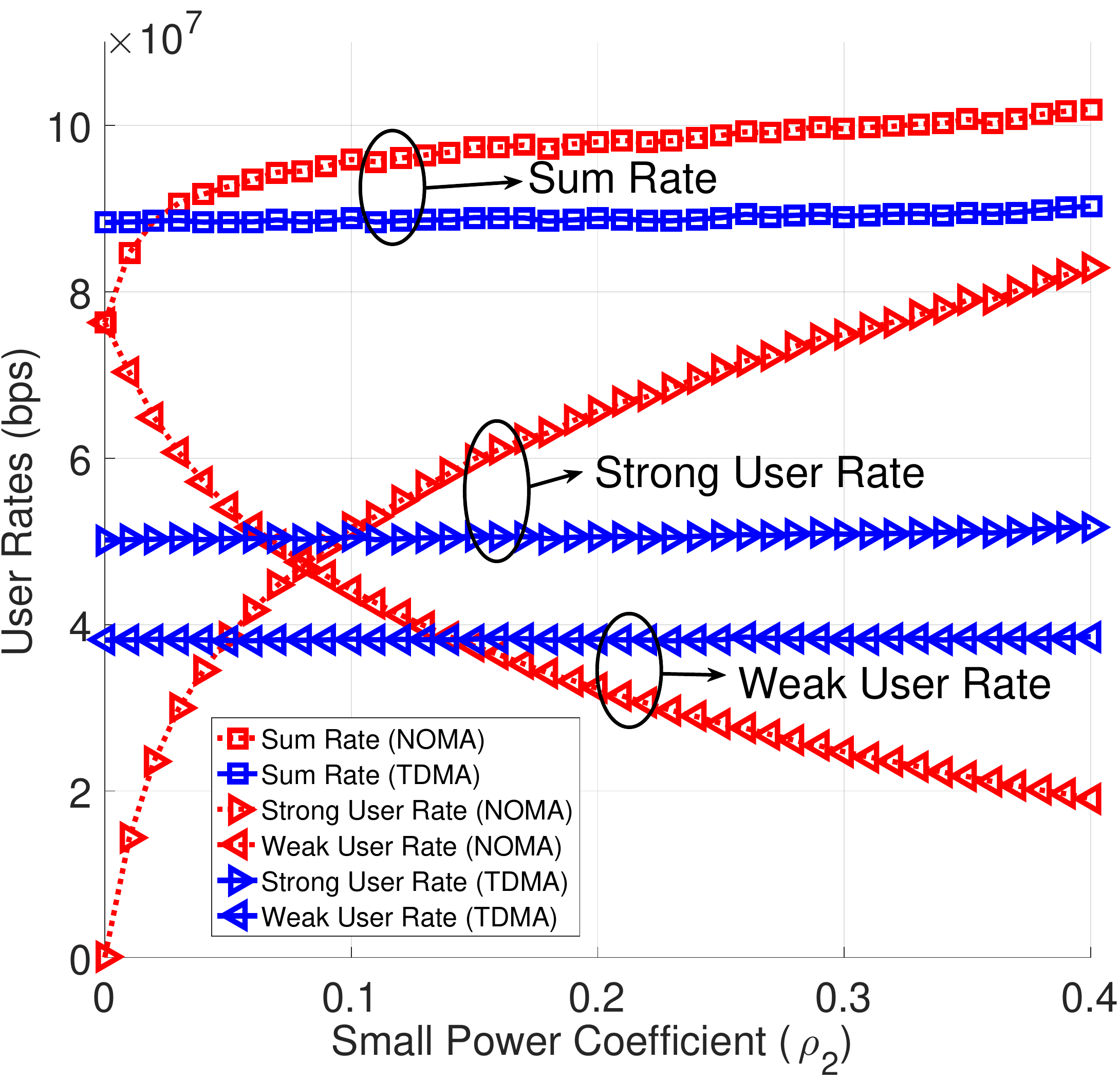}
		\label{VaryingCoeff}
	}~~
	\subfigure[The ratio of NOMA user pairs achieving SINR threshold for varying threshold and different $\rho_2$ values.]{
		\includegraphics[width = 3 in] {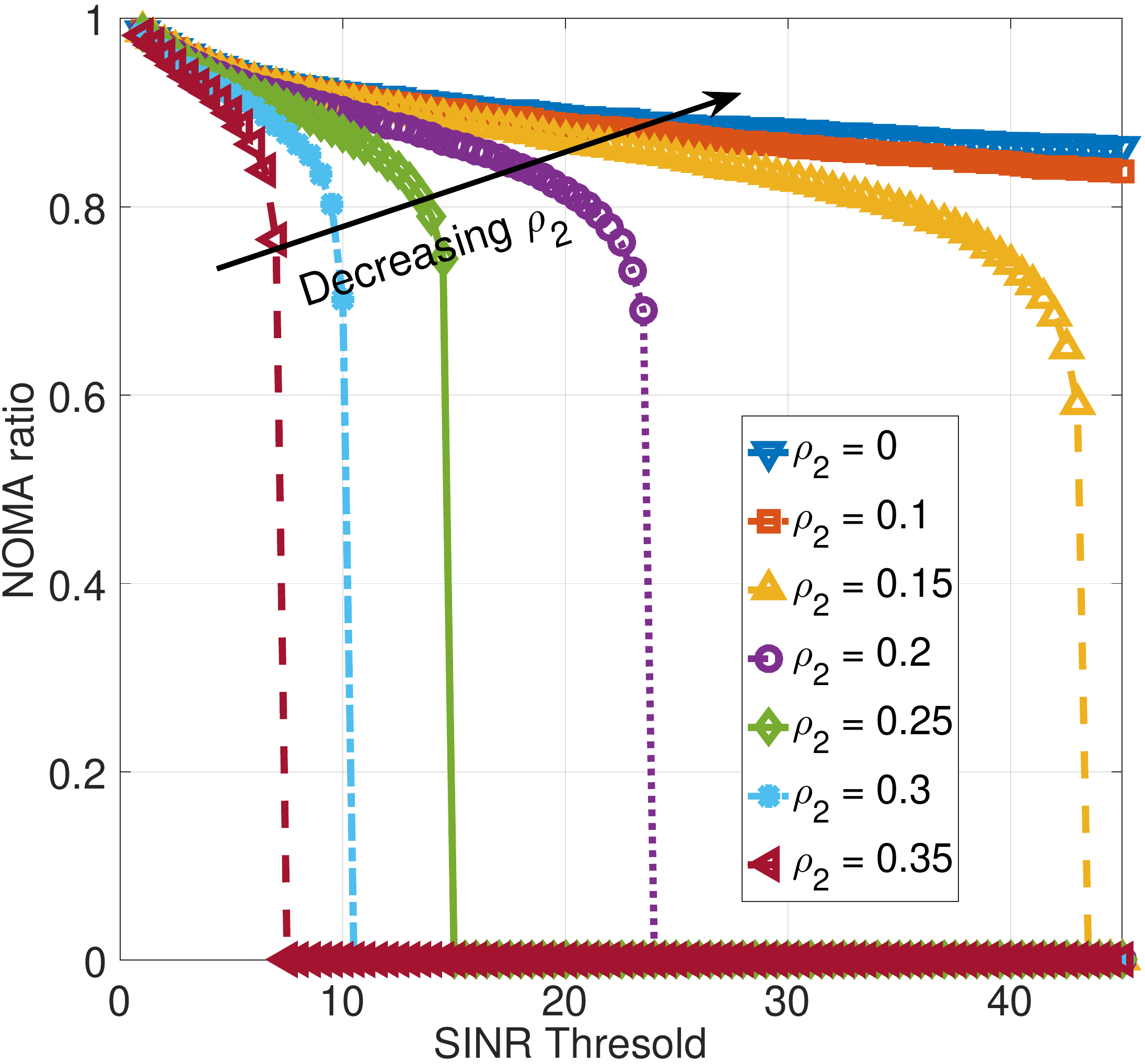}
		\label{VaryingThreshold}
	}\vspace{-1mm}
    \caption{ NOMA user rates and the ratio of users achieving NOMA threshold. }\vspace{-5mm}
\end{figure}

Another parameter for designing NOMA is the SINR threshold that needs to be achieved to implement NOMA. The SINR for the second user to decode the signal of the first user is given in \eqref{SINR_Thre}. A threshold $\xi_{2\rightarrow1} \geq \xi_1^{*}$ should be chosen as a design parameter for NOMA preference over TDMA. Fig.~\ref{VaryingThreshold} shows the ratio of NOMA users that achieves the threshold for different $\rho_2$ and $\xi_1^{*}$ levels. For small values of $\rho_2$, the SINR threshold is easily achieved even for a high threshold. Small $\rho_2$ causes the $\xi_{2\rightarrow1}$ to be larger, which provides less error probability for successive interference cancellation. For larger values of $\rho_2$, the SINR threshold should be decreased to allow NOMA. This also increases the risk of erroneous interference cancellation. Fig.~\ref{VaryingCoeff} suggests that the sum rate can be increased by increasing $\rho_2$. However, Fig.~\ref{VaryingThreshold} also shows that high $\rho_2$ may cause most users to not to use NOMA, which may diminish the sum rate gain. Overall, making $\rho_2$ smaller, or making the power coefficient of the users as distinct as possible, provides fairer rate increase of users and decreases the probability of erroneous interference cancellation.

In Fig.~\ref{CDF_NOMA}, we present the CDF of NOMA user rates with optimized coefficients as in \eqref{NOMA_Opt3}. The black line with circle markers shows the data rates of user pairs in case these users are served with TDMA. The solid line with lower data rate is for the weak user, and the dashed line with higher data rate is for the strong user. On the x-axis, between $10^7$ and $10^8$, each vertical line means a $10^7$ bps data rate. There is about a 10 Mbps data rate difference between weak and strong TDMA users. The red line with triangle markers shows the data rates for the same users when they are served by NOMA, and the NOMA coefficients are calculated to maximize the sum rate using \eqref{NOMA_Opt3}. The strong user rate that is shown with dashed lines has a significant rate gain compared to TDMA in the high data rate region, which is the upper parts of the line. The weak user does not have better data rates than TDMA in the high rate region, but it is better at low rate region. Overall, the NOMA provides gain for some users, but it decreases the data rates for some other users, which might be both weak or strong user in the pair.

\begin{figure}[tp]
	\centering
	\vspace{-1mm}
	\subfigure[The CDF of individual user rates. WU stands for weak user and SU stands for strong user.]{
		\includegraphics[width = 3 in] {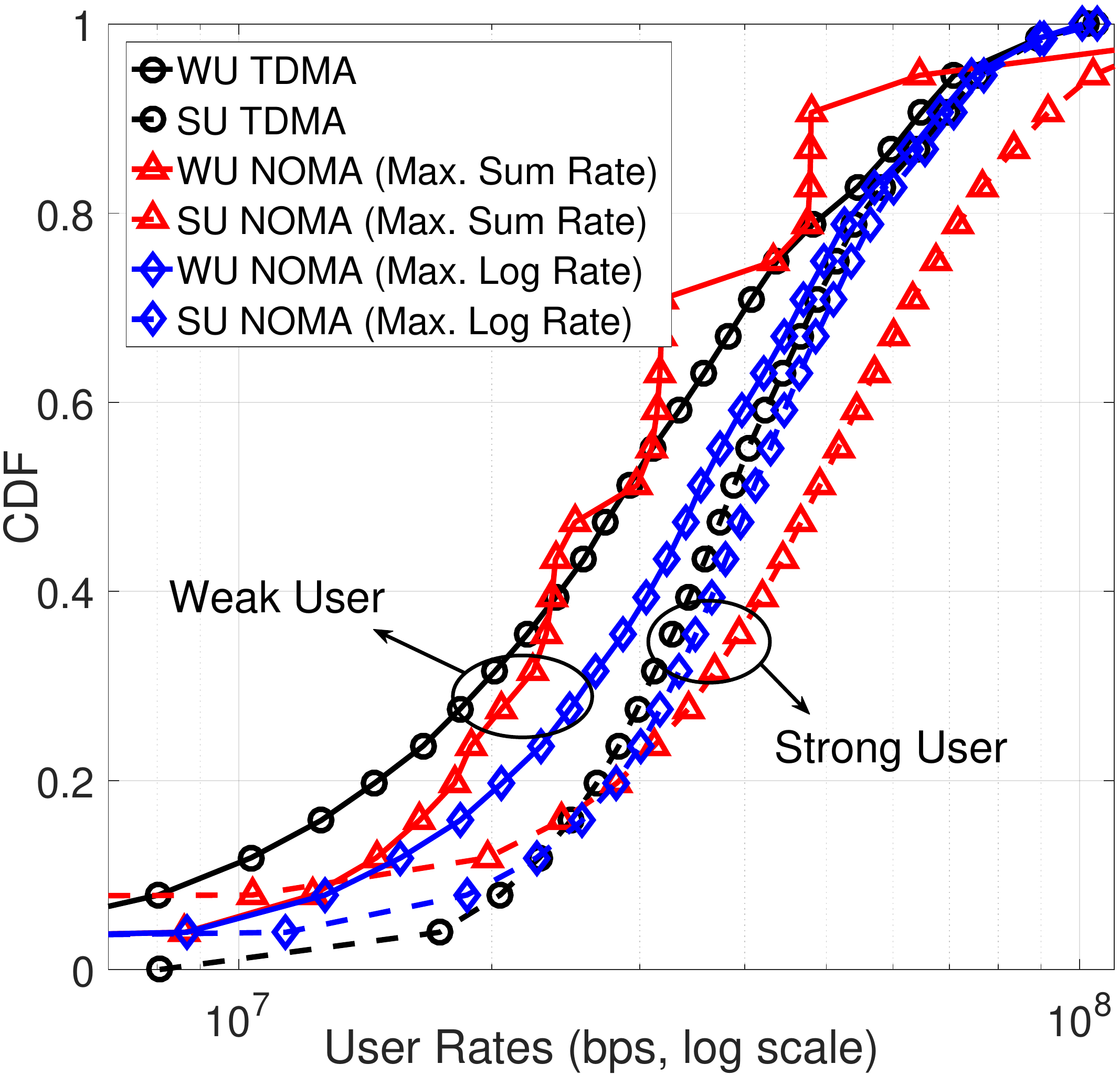}
		\label{CDF_NOMA}
	}~~
	\subfigure[The CDF of sum rates.]{
		\includegraphics[width = 3 in] {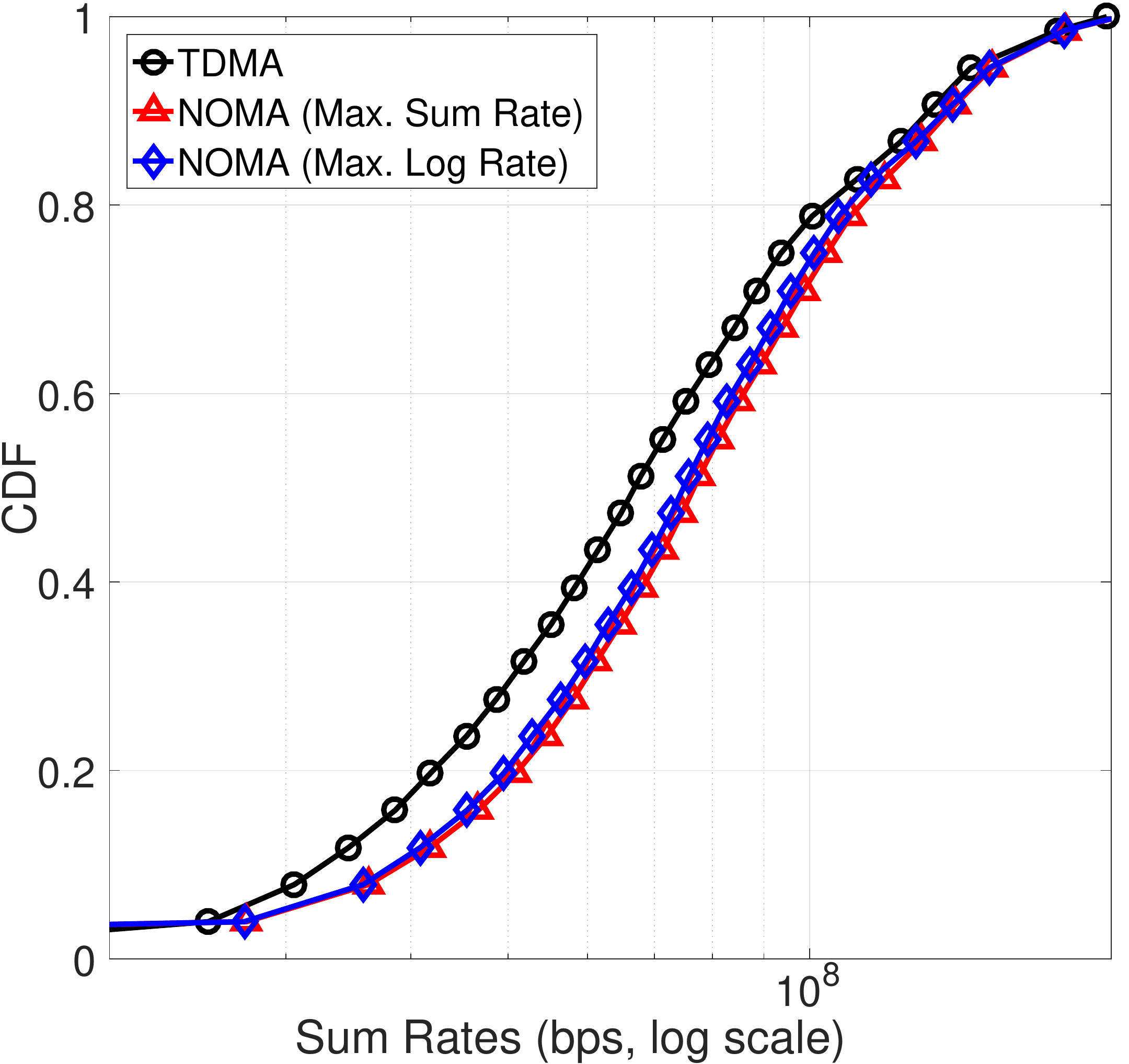}
		\label{CDF_NOMA_Sum_Rates}
	}\vspace{-1mm}
    \caption{The CDF of user rates for users pairs that are eligible for NOMA. }\vspace{-5mm}
\end{figure}

In Fig.~\ref{CDF_NOMA}, the blue line with diamond markers show the data rates for the same users when they are served by NOMA, and the NOMA coefficients are calculated to maximize the sum of the logarithm of the rates. In this case, the weak user has a significant data rate gain compared to TDMA in all regions. The strong user has a small data rate gain compared to TDMA in most regions. Only the bottom $18\%$ of the strong users has a rate loss compared to TDMA, but they still do better than weak users. Overall, this NOMA scheme provides a significant data rate gain for all weak users, and a slight data rate gain for most strong users, with a small data rate loss for some strong users. 

In Fig.~\ref{CDF_NOMA_Sum_Rates}, the CDF of the sum rates are shown for the same users in Fig.~\ref{CDF_NOMA}. Both NOMA schemes provide a significant sum rate gain over TDMA. The sum rate difference between TDMA and NOMA is about 10 Mbps for most users. The NOMA coefficients that maximize sum rate provides a slight sum rate gain over the coefficients that maximize the sum of the logarithm of the rates. 

\vspace{-2mm}

\section{Conclusion} \label{Conc}
In this paper, we study the optimal beam steering parameters for VLC when there are more users than the steerable components. We find the optimal steering angles and LED directivity for a single LED and multiple users. The results show that steering VLC beams and changing the directivity can improve the user rates significantly. Although serving a single user maximizes the user rates, multiple users can also be served using a single steerable beam with a significant sum rate gain over no steering scheme. We also propose a method for decreasing the search space, thus the computation time of the optimization solution. This method decreases the search space to 1\%-90\% depending on user distribution. In case of a multiple steerable beam setting, we cluster users and serve each cluster with a separate beam. This setting allows higher data rates by clustering close users together and providing more accurate steering. Additionally, we optimize the transmit power of each beam to increase the sum rate or proportionally fair sum rate. With the clustering and power optimization, the sum rate can reach ten times of the sum rate of no steering scheme. Finally, we propose a user clustering and NOMA scheme to utilize the space diversity of the users and further increase the data rates. NOMA can provide an additional 10 Mbps rate gain for two users that are paired together.

\appendices
\section{Proof of Proposition~\ref{prop1}} \label{Appendix1}

Consider the LOS channel gain between the transmitter LED and the $k$-th receiver in \eqref{h_first}. The only parameter in \eqref{h_first} affected by the orientation is $\phi_k$, which is the angle between the LED orientation and the vector $\textbf{v}_k$ from LED to the $k$-th receiver. When the LED points to the receiver, $\phi_k = 0$, and the channel gain is maximized. When $\phi_k$ increases, the channel gain decreases. 

Now consider that we have two users in the system as shown in Fig.~\ref{ProofFigure}(a), and denote them user-1 and user-2. Let the line segment between the locations of two users is denoted by $\mathcal{K}$. Let $\phi_{1-2}$ be the angle between $\textbf{v}_1$ and $\textbf{v}_2$, the vectors towards user-1 and user-2 from the LED, respectively. Note that, the $\phi_{1-2}$ is independent of the orientation of the beam. First, assume that the LED is either pointing to user-1 or user-2, or somewhere on $\mathcal{K}$. In this case $\phi_{1-2} = \phi_1 + \phi_2$. When the LED is steered towards user-1, the $\phi_1 = 0$, and $\phi_2 = \phi_{1-2}$. 

Now assume that the LED is steered to a point not on the line segment $\mathcal{K}$. This steering is not Pareto efficient because $ \phi_1 + \phi_2 > \phi_{1-2}$. We can find a steering orientation with smaller $\phi_1$ and $\phi_2$ if the LED is steered towards $\mathcal{K}$. In Fig.~\ref{ProofFigure}(a) we illustrate a scenario where the LED is steered towards point A, which is not on $\mathcal{K}$. If the LED is steered towards point B (the projection of point A onto $\mathcal{K}$) instead of point A, both angles $\phi_1$ and $\phi_2$ will decrease, which means (based on \eqref{h_first}) that both users will have higher channel gains. Therefore steering the LED towards point A cannot be optimal. The optimal steering angles that maximize \eqref{betaOpt} point towards either one of the two users, or somewhere on $\mathcal{K}$.
\hfill\IEEEQEDhere

\begin{figure}[tp]
	\centering
	\vspace{-1mm}
	\subfigure[Two users.]{
		\includegraphics[width=2in] {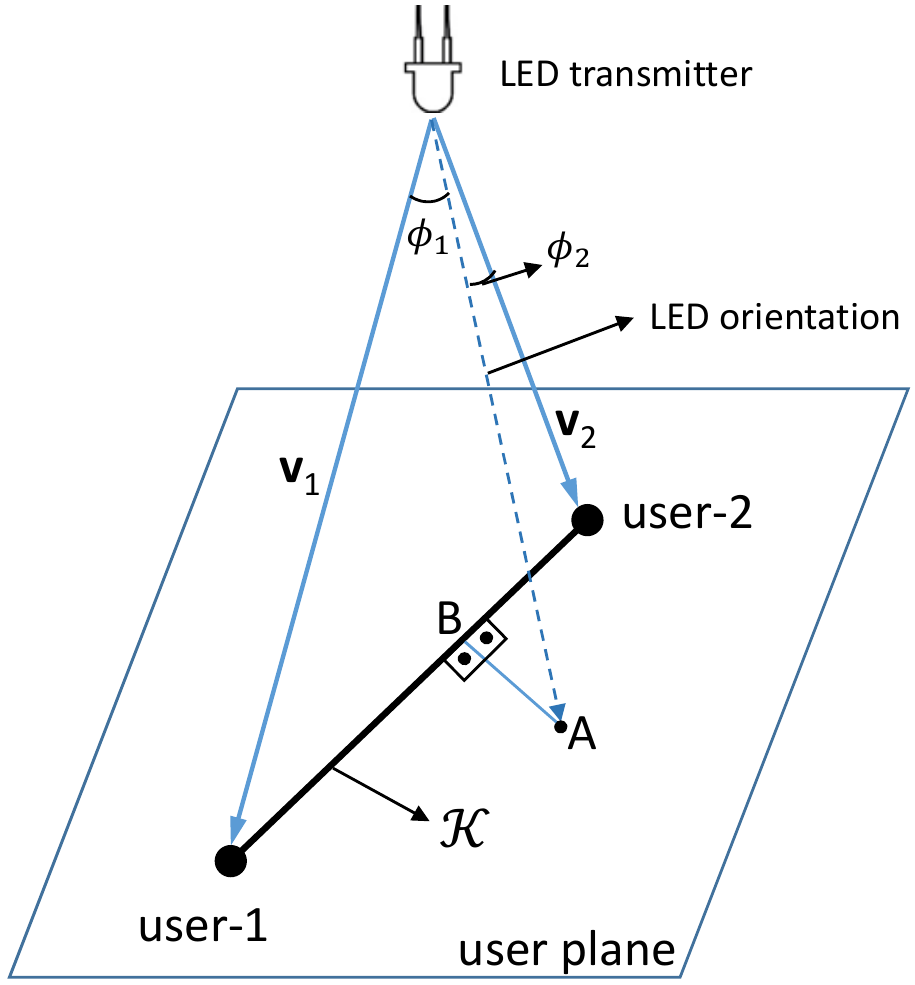}
		\label{ProofTwoUsers}
	}
	\subfigure[Three users.]{
		\includegraphics[width=2.3in] {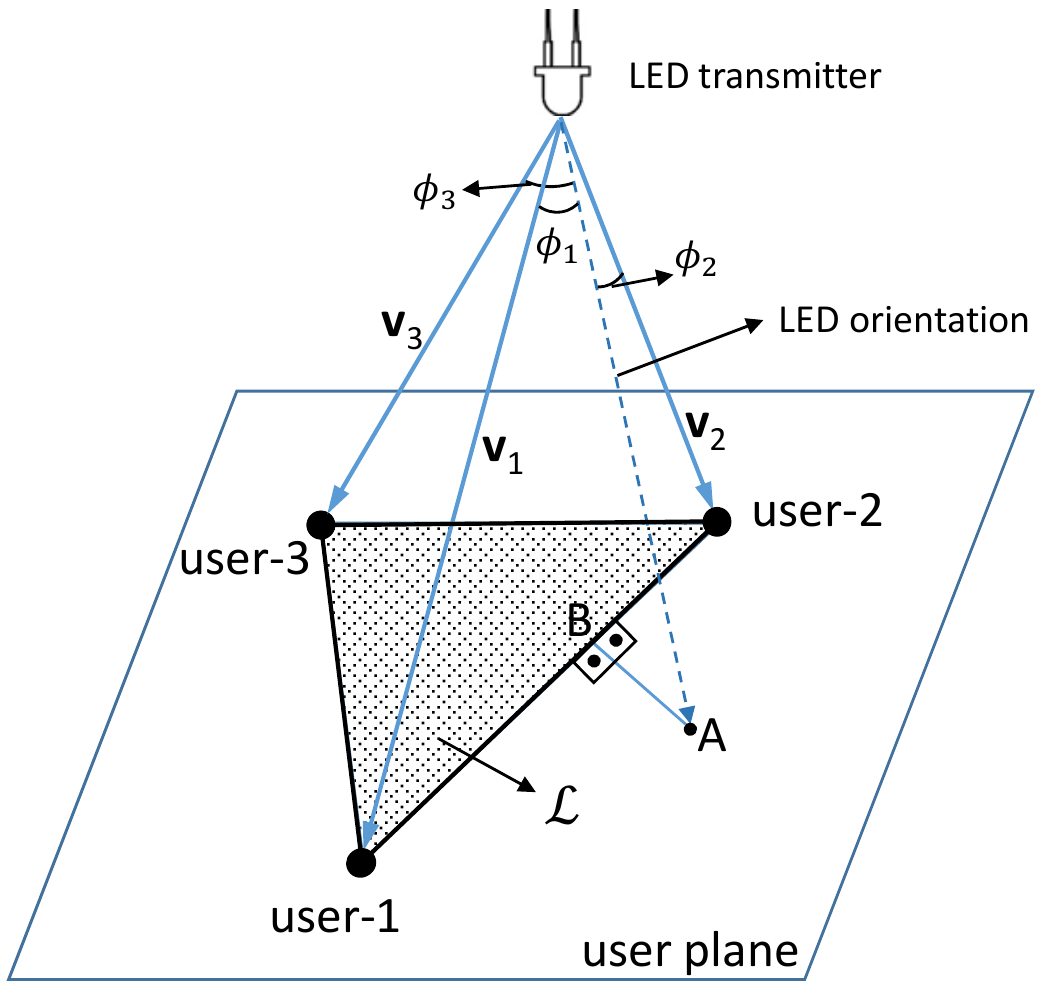}
		\label{ProofThreeUsers}
	}\vspace{-2mm}
    \caption{Steering the LED towards users in two or three user scenarios. In both scenarios, steering the LED towards point B instead of point A provides higher LOS channel gain to all users in the system.}
	\label{ProofFigure}\vspace{-3mm}
\end{figure}

\section{Proof of Proposition~\ref{prop2}} \label{Appendix2}

In case there are three users that are not on the same line, the optimal steering orientation of the LED that is described in \eqref{betaOpt} has to point somewhere within the triangle defined by the user locations. Let us denote the triangle $\mathcal{L}$. For any orientation of the LED that does not point to $\mathcal{L}$, we can find some orientation that points to it and has smaller $\phi_1$, $\phi_2$, and $\phi_3$. To prove that, assume the LED is steered towards point A, which is not on $\mathcal{L}$ but on the plane that includes $\mathcal{L}$. Now change the intersection point to the closest point to A within $\mathcal{L}$, and steer the LED towards that point. All three angles $\phi_1$, $\phi_2$, and $\phi_3$ will decrease, which means the previous steering angle was not Pareto efficient. It can be seen in Fig.~\ref{ProofFigure}(b) that, steering the LED towards point B provides all three angles to be smaller compared to point A. If the LED orientation does not intersect the plane at all, the steering is not good and should be changed towards the user locations. Similarly, if there are more than 3 users that are on the same plane, the optimal steering angle points to somewhere within the convex hull of location points. It can be shown using the same method applied to two user and three user cases.
\hfill\IEEEQEDhere

\bibliographystyle{IEEEtran} 
\bibliography{new}

\end{document}